%% file: ms.tex
\documentclass[preprint,1p]{elsarticle}

\usepackage{lineno,hyperref}

\pdfoutput=1

\modulolinenumbers[5]

\journal{Information Systems}
\newcommand{\metodo}{Lathe}

\usepackage[english]{babel}      
\usepackage[utf8]{inputenc}
\usepackage{subcaption}
\usepackage[T1]{fontenc}
\usepackage{type1ec}
\usepackage{graphicx} 
\usepackage{url}
\usepackage{rotating}
\usepackage{amsmath}
\usepackage{amsfonts}
\usepackage{adjustbox}
\usepackage{multirow}

    \usepackage{tabularx}
    \newcolumntype{Y}{>{\centering\arraybackslash}X}

    \usepackage{xfrac}

\allowdisplaybreaks 

\usepackage{listings} 
\usepackage{textcomp} 
\usepackage{rotating} 

\makeatletter
\newcommand{\subalign}[1]{%
	\vcenter{%
		\Let@ \restore@math@cr \default@tag
		\baselineskip\fontdimen10 \scriptfont\tw@
		\advance\baselineskip\fontdimen12 \scriptfont\tw@
		\lineskip\thr@@\fontdimen8 \scriptfont\thr@@
		\lineskiplimit\lineskip
		\ialign{\hfil$\m@th\scriptstyle##$&$\m@th\scriptstyle{}##$\hfil\crcr
			#1\crcr
		}%
	}%
}
\makeatother

\usepackage{array}
\newcolumntype{L}[1]{>{\raggedright\let\newline\\\arraybackslash\hspace{0pt}}m{#1}}
\newcolumntype{C}[1]{>{\centering\let\newline\\\arraybackslash\hspace{0pt}}m{#1}}
\newcolumntype{R}[1]{>{\raggedleft\let\newline\\\arraybackslash\hspace{0pt}}m{#1}}

\usepackage[g]{esvect} 

\usepackage[linesnumbered,ruled,
			vlined,
			nofillcomment]{algorithm2e} 
\SetKw{KwAnd}{and}
\SetKwHangingKw{Let}{let}
\SetArgSty{textnormal}
\DontPrintSemicolon

\usepackage{afterpage}
\usepackage{lscape}

\usepackage{mathtools} 

\usepackage{microtype}
\usepackage{times,courier}
\renewcommand{\subparagraph}{\paragraph}
\usepackage{float}
\newfloat{copyrightbox}{thp}{lop}
\usepackage[skip=2pt]{caption}
\usepackage{mdwlist}
\usepackage{amssymb} 

\usepackage{amsthm}
\newtheoremstyle{noind}
{5pt}
{5pt}
{\itshape}
{}
{\bfseries}
{.}
{2pt} 
{}
\theoremstyle{noind}

\newtheorem{definition}{Definition}

\newtheorem{example}{{Example}}

    \usepackage{thmtools, thm-restate}

\newif\ifReview
\Reviewtrue 
\usepackage[normalem]{ulem}
\usepackage[dvipsnames]{xcolor}
\usepackage{fancybox}

\newcommand{\red}[1]{\textcolor{red}{#1}}

\ifReview
    \newcommand{\mycomment}[2]{\newline\noindent\shadowbox{\begin{minipage}{\dimexpr0.98\textwidth-\shadowsize-2\fboxrule-2\fboxsep}\textbf{#1: #2}\end{minipage}}}
    \newcommand{\add}[1]{\textcolor{blue}{#1}}

    \newcommand{\rem}[1]{\red{\sout{#1}}}
\else
    \newcommand{\mycomment}[2]{}
    \newcommand{\add}[1]{{#1}}
    
    \newcommand{\tirar}[1]{}
    \newcommand{\rem}[1]{}
\fi

\usepackage{tikz}
\usetikzlibrary{positioning}
\tikzset{main node/.style={text centered},}

\newcommand*\circled[1]{\tikz[baseline=(char.base)]{
            \node[shape=circle,draw,inner sep=2pt] (char) {#1};}}
\usetikzlibrary{shapes.geometric} 

    \newcounter{tuplesno}
    \newcommand{\deftuple}[1]{
        \refstepcounter{tuplesno}
        \label{t:#1}
        \reftuple{#1}
    }
    \newcommand{\reftuple}[1]{$t_{\reftupleno{#1}}$}
    \newcommand{\reftupleno}[1]{\ref*{t:#1}}
    \newcommand{\mreftuple}[1]{t_{\reftupleno{#1}}}

    \usetikzlibrary{calc}

\usepackage{booktabs}

\begin{document}
	
	\begin{frontmatter}
		
		\title{Supporting Schema References in Keyword Queries \\ over Relational Databases}
		
		\author{Paulo Martins}
		\author{Altigran da Silva}
		\author{João Cavalcanti}
	    \author{Edleno de Moura}
		\address{Institute of Computing, Federal University of Amazonas, Manaus, Brazil
		\{paulo.martins,alti,john,edleno\}@icomp.ufam.edu.br}

		\begin{abstract}
			  Relational Keyword Search (R-KwS) systems enable naive/informal users to explore and retrieve information from relational databases without knowing schema details or query languages. These systems take the keywords from the input query, locate the elements of the target database that correspond to these keywords, and look for ways to “connect” these elements using information on referential integrity constraints, i.e., key/foreign key pairs. Although several such systems have been proposed in the literature, most of them only support queries whose keywords refer to the contents of the target database and just very few support queries in which keywords refer to elements of the database schema. This paper proposes LATHE, a novel R-KwS designed to support such queries. To this end, in our work, we first generalize the well-known concepts of Query Matches (QMs) and Candidate Joining Networks (CJNs) to handle keywords referring to schema elements and propose new algorithms to generate them. Then, we introduce an approach to automatically select the CJNs that are more likely to represent the user intent when issuing a keyword query. This approach includes two major innovations: a ranking algorithm for selecting better QMs, yielding the generation of fewer but better CJNs, and an eager evaluation strategy for pruning void useless CJNs. We present a comprehensive set of experiments performed with query sets and datasets previously used in experiments with state-of-the-art R-KwS systems and methods. Our results indicate that LATHE can handle a wider variety of keyword queries while remaining highly effective, even for large databases with intricate schemas.
		\end{abstract}
		
		
	\end{frontmatter}

	\section{Introduction}\label{chap:intro}
	\input{sections/1-introduction}
	
	\section{Background and Related Work}\label{chap:related-work}
	\input{sections/2-related-work}
	
	\section{ {\metodo}  Overview}\label{chap:overview} 
	\input{sections/3-overview}
	
	\section{Keyword Matching}\label{chap:keyword-matching} 
	\input{sections/4-keyword-matching}

	\section{Query Matching}\label{chap:query-matching} 
	\input{sections/5-query-matching}
	
	    \subsection{Query Matches Ranking} \label{sec:query-match-ranking}
	    \input{sections/5b-qm-ranking}	
	
	\section{Candidate Joining Networks}\label{chap:candidate-network} 
	\input{sections/6-candidate-network-generation}
	
	\section{Experiments}\label{chap:experiments} 
	\input{sections/7-experiments}

	\section{Conclusions}\label{chap:future-work}
	\input{sections/8-conclusions}

	\section{Acknowledgments}\label{chap:acknowledgments}
	This research, according for in Article 48 of Decree nº 6.008/2006, was funded by Samsung Electronics of Amazonia Ltda, under the terms of Federal Law nº 8.387/1991, through agreement nº 003, signed with  ICOMP/UFAM.
	It was also supported by the Brazilian funding agency FAPEAM-POSGRAD 2020 (Resolution 002/2020), the Coordination for the Improvement of Higher Education Personnel-Brazil (CAPES) financial code 001, by projects SocSens (CAPES/ PCGI 88887.130299/2017-01), MMBIAS (FAPESP MCTIC/CGI, 2020/05173-4), and authors’ individual grants from CNPq.

	\input{sections/9-apendices}

\end{document}

%% file: sections/1-introduction.tex

\textit{Keyword Search over Relational Databases} (R-KwS) enables naive/informal users to retrieve information from relational databases (DBs) without any knowledge about schema details or query languages. The success of search engines shows that untrained users are at ease using keyword search to find information of interest. 
However, enabling users to search relational DBs using keyword queries is a challenging task because the information sought frequently spans multiple relations and attributes, depending on the schema design of the underlying DB. As a result, R-KwS systems face the challenge of automatically determining which pieces of information to retrieve from the database and how to connect them to provide a relevant answer to the user.

In general, the keywords from a query may refer to both database values, such as tuples containing these keywords, and schema elements, such as relation and attribute names. For instance, consider the query 
$``will\ smith\ films"$
over a database on movies. The keywords $``will"$ and $``smith"$ may refer to values of person names. 
The keyword $``films"$ on the other hand is more likely to refer to a schema element, the name of the relation about movies.
Although a significant number of query keywords correspond to schema references \cite{Bergamaschi@SIGMOD11Keymantic}, the majority of previous work on R-KwS systems in the literature does not support references to schema information such as the one in the query above. 
As a result, given a query, they will search for attributes whose tuples include the keyword $``films"$, which is unlikely to yield a useful answer for the user.

\newcommand{\ftnlathe}{\footnote{The name {\metodo} refers to the fact that our system assigns a structure or form to an unstructured keyword-based query}}

\newcommand{\ftncjn}{\footnote{Most of the previous work uses the term \emph{Candidate Networks} instead. Here, we use \emph{Candidate Joining Networks} because we consider it more meaningful.}}

In this work, we study new techniques for supporting schema references in keyword queries over relational databases. Specifically, we propose {\metodo}\ftnlathe, a new R-KwS system to generate a suitable SQL query from a keyword query, considering that keywords may refer either to instance values or to database schema elements, i.e., relations and attributes. 
{\metodo} follows the \textit{Schema Graph} approach for R-KwS systems \cite{Oliveira@ICDE18MatCNGen,Coffman@CIKM10Framework}. Given a keyword query, this approach consists of generating relational algebra expressions called \textit{Candidate Joining Networks}{\ftncjn} (CJNs), which are likely to express user intent when formulating the original query. 
The generated CJNs are \textit{evaluated}, that is, they are translated into SQL queries and executed by a DBMS, resulting in several \emph{Joining Networks of Tuples} (JNTs) which are collected and supplied to the user. 

In the literature, the most well-known algorithm for CJN Generation is CNGen, which was first presented in the system DISCOVER \cite{Hristidis@VLDB02DISCOVER}, but was adopted by most R-KwS systems \cite{Agrawal@ICDE02DBXplorer,Hristidis@VLDB03Efficient,Luo@SIGMOD07Spark,Coffman@KEYS10CD}. Despite the possibly large number of CJNs, most works in the literature focused on improving CN Evaluation and ranking of JNTs instead. Specifically, DISCOVER-II \cite{Hristidis@VLDB03Efficient}, SPARK \cite{Luo@SIGMOD07Spark}, and CD \cite{Coffman@KEYS10CD} used information retrieval (IR) style score functions to rank the top-K JNTs. KwS-F \cite{Baid@VLDB10KwSF} imposed a time limit for CJN evaluation, returning potentially partial results as well as a summary of the CJNs that have yet to be evaluated.
Later, CNRank~\cite{Oliveira@ICDE15CNRank} introduces a CJN ranking, requiring only the top-ranked CJNs to be evaluated. 
MatCNGen~\cite{Oliveira@ICDE18MatCNGen,Oliveira@TKDE20} proposed a novel method for generating CJNs that efficiently enumerated the possible matches for the query in the DB. These \textit{Query Matches} (QMs) are then used to guide the CJN generation process, greatly decreasing the number of generated CJNs and improving the performance of CJN evaluation.

Among the methods based on the Schema Graph approach, {\metodo} is, to the best of our knowledge, the first method to address the problem of generating and ranking CJNs considering queries with keywords that can refer to either schema elements or attribute values. 
We revisited and generalized concepts introduced in previous approaches~\cite{Hristidis@VLDB02DISCOVER,Oliveira@ICDE15CNRank,Oliveira@ICDE18MatCNGen,Oliveira@TKDE20}, such as tuples-sets, QMs, and the CJNs themselves, to enable schema references.
In addition, we proposed a more effective approach to CJN Generation that included two major innovations: 
QM ranking and Eager CJN Evaluation.
{\metodo} roughly matches keywords to the values of the attributes or to schema elements such as names of attributes and relations. Next, the system combines the keyword matches into QMs that cover all the keywords from the query. The QMs are ranked and only the most relevant ones are used to generate CJNs. The CJN generation explores the primary key/foreign key relationships to connect all the elements of the QMs. 
In addition, {\metodo} employs an eager CJN evaluation strategy, which ensures that all CJNs generated will yield non-empty results when evaluated. The CJNs are then ranked and evaluated. Finally, the CJN evaluation results are delivered to the user.
Unlike the previous methods, {\metodo} provides the user with the most relevant answer without relying on JNTs rankings. This is due to the effective rankings of QMs and CJNs that we propose, which are absent in the majority of previous work.

We performed several experiments to assess the effectiveness and efficiency of {\metodo}. 
First we compared the quality of its results with those obtained with several previous R-KwS systems, including the state-of-the-art QUEST~\cite{Bergamaschi@VLDBDEMO13_QUEST} system using a benchmark proposed by Coffman \& Weaver \cite{Coffman@CIKM10Framework}.
Second we assessed the quality of our ranking of QMs. The ranking of CJNs was then evaluated by comparing different configurations in terms of the number of QMs, the number of CJNs generated per QM, and the use of the eager evaluation strategy. 
Finally, we assessed the performance of each phase of {\metodo}, as well as the trade off between quality and performance of various system configurations. 
{\metodo} achieved better results than all of the R-KwS systems tested in our experiments. 
Also, our results indicate that the ranking of QMs and the eager CJN evaluation greatly improved the quality of the CJN generation.

Our key contributions are: 
(i) a novel method for generating and ranking CJNs with support for keywords referring to schema elements;
(ii) a novel algorithm for ranking QMs, which avoids the processing of less likely answers to a keyword query;
(iii) an eager CJN evaluation for discarding spurious CJNs;
(iv) a simple and yet effective ranking of CJNs which exploits the ranking of QMs. 

The remainder of this paper is organized as follows: 
Section~\ref{chap:related-work} reviews the related literature on relational keywords search systems based on schema graphs and support to schema references. 
Section~\ref{chap:overview} summarizes all of the phases of our method, which are discussed in detail in 
Sections~\ref{chap:keyword-matching}-\ref{chap:candidate-network}. Section~\ref{chap:experiments} summarizes the findings of the  experiments we conducted.
Finally, Section~\ref{chap:future-work} summarizes the findings and outlines our plans for future.

%% file: sections/2-related-work.tex

In this section, we discuss the background and related work on keyword search systems over relational databases and on supporting schema references in such systems. For a more comprehensive view of the state-of-the-art in keyword-based and natural language queries over databases, we refer the interested reader to a recent survey~\cite{Affolter@VLDBJ2019_SurveyNLIDB}.

\subsection{Relational Keyword Search Systems}

Current R-KwS systems fall in one of two distinct categories: systems based on \emph{Schema Graphs} and systems based on 
\emph{Instance Graphs}.  Systems in the first category are based on the concept of \emph{Candidate Joining Networks} (CJNs), which are networks of joined relations that are used to generate SQL queries and whose evaluation return several \emph{Joining Networks of Tuples} (JNTs) which are collected and supplied to the user. 
This method was proposed in DISCOVER \cite{Hristidis@VLDB02DISCOVER} and DBXplorer \cite{Agrawal@ICDE02DBXplorer}, and it was later adopted by several other systems, including DISCOVER-II \cite{Hristidis@VLDB03Efficient}, SPARK \cite{Luo@SIGMOD07Spark}, CD \cite{Coffman@KEYS10CD}, KwS-F \cite{Baid@VLDB10KwSF}, CNRank \cite{Oliveira@ICDE15CNRank}, and MatCNGen \cite{Oliveira@ICDE18MatCNGen,Oliveira@TKDE20}. Systems in this category make use of the underlying basic functionality of the RDBMS by generating appropriate SQL queries to retrieve answers to keyword queries posed by users.

Systems in the second category are based on a structure called \emph{Instance Graph}, whose nodes represent tuples associated with the keywords they contain, and the edges connect these tuples based on referential integrity constraints. 
BANKS \cite{Aditya@VLDB02BANKS}, BANKS-II \cite{Kacholia@VLDB05Bidirectional}, BLINKS \cite{He@SIGMOD07BLINKS} and, Effective \cite{Liu@SIGMOD06Effective} use this approach to compute keyword queries results by finding subtrees in a data graph that minimizes the distance between nodes matching the given keywords. 
These systems typically generate the query answer in a single phase that combines the tuple retrieval task and the answer schema extraction. However, the Instance Graph approach requires a materialization of the DB and requests a higher computational cost to deliver answers to the user. Furthermore, the important structural information provided by the database schema is ignored, once the data graph has been built.

\subsection{R-KwS Systems based on Schema Graphs}

In our research, we focus on systems based on Schema Graphs, since we assume that the data we want to query are stored in a relational database and we want to use an RDBMS capable of processing SQL queries. Also, our work expands on the concepts and terminology introduced in DISCOVER \cite{Hristidis@VLDB02DISCOVER,Hristidis@VLDB03Efficient} and expanded in CNRank \cite{Oliveira@ICDE15CNRank} and MatCNGen \cite{Oliveira@ICDE18MatCNGen,Oliveira@TKDE20}. This formal framework is used and expanded to handle keyword queries that may refer to attribute values or to database schema elements. As a result, we can inherit and maintain all guarantees regarding the generation of minimal, complete, sound, and meaningful CJNs.

The best-known algorithm for CJN Generation is CNGen, which was introduced in DISCOVER~\cite{Hristidis@VLDB02DISCOVER} but was later adopted as a default in most of the R-KwS systems proposed in the literature \cite{Agrawal@ICDE02DBXplorer,Hristidis@VLDB03Efficient,Luo@SIGMOD07Spark,Coffman@KEYS10CD}. To generate a complete, non-redundant set of CJNs, this algorithm employs a Breadth-First Search approach \cite{Cormen@09Algorithms}. As a result, CNGen frequently generates a large number of CJNs, resulting in a costly CJN generation and evaluation process.

Initially, most of the subsequent work focused on the CJN evaluation only. Specifically, as many CJNs were generated by CNGen that should be evaluated, producing a larger number of JNTs,
such systems as DISCOVER-II \cite{Hristidis@VLDB03Efficient}, SPARK \cite{Luo@SIGMOD07Spark}, and CD \cite{Coffman@KEYS10CD} introduced algorithms for ranking JNTs using IR style score functions.

KwS-F~\cite{Baid@VLDB10KwSF} addressed the efficiency and scalability problems in CJN evaluation in a different way. Their approach consists of two steps. First, a limit is imposed on the time the system spends evaluating CJNs. After this limit is reached, the system must return the (possibly partial) top-K JNTs. 
Second, if there are any CJNs that have yet to be evaluated, they are presented to the user in the form of query forms, from which the user can choose
one and the system will evaluate the corresponding CJN. 

CNRank~\cite{Oliveira@ICDE15CNRank} proposed a method for lowering the cost of CJN evaluation by ranking them based on the likelihood that they will provide relevant answers to the user.  Specifically, CNRank presented a probabilistic ranking model that uses a \emph{Bayesian Belief Network} to estimate the relevance of a CJN given the current state of the underlying database. A score is assigned to each generated CJN, so that only a few CJNs with the highest scores 
need to be evaluated.

MatCNGen~\cite{Oliveira@ICDE18MatCNGen,Oliveira@TKDE20} introduced a match-based approach for generating CJNs. The system enumerates the possible ways which the query keywords can be matched in the DB beforehand,
to generate query answers. MatCNGen then generates a single CJN, for each of these QMS, drastically reducing the time required to generate CJNs. 
Furthermore, because the system assumes that answers must contain all of the query keywords, each keyword must appear in at least one element of a CJN. As a result of the generation process avoiding generating too many keyword occurrence combinations, a smaller but better set of CJNs is generated. 

Lastly, Coffman \& Weaver~\cite{Coffman@CIKM10Framework} proposed a framework for evaluating R-KwS systems and reported experimental results over three representative standardized datasets they built, namely MONDIAL, IMDb, and Wikipedia, along with respective query workloads. The authors compare nine R-KwS systems, assessing their effectiveness and performance in a variety of ways. The resources of this framework were also used in the experiments  of several other studies
on R-KwS systems~\cite{Luo@SIGMOD07Spark,Oliveira@ICDE15CNRank,Oliveira@ICDE18MatCNGen,Oliveira@TKDE20,Coffman@TKDE12Evaluation}.

\subsection{Support to Schema References in R-KwS}

Overall there are few systems in the literature that support schema references in keywords queries. One of the first such systems was BANKS~\cite{Bhalotia@ICDE02BANKS}, a R-KwS system based on Instance Graphs. However, hence the query evaluation with keywords matching metadata can be relatively slow, since a large number of tuples may be defined to be relevant to the keyword.

Support for schema references in keyword queries was extensively addressed by Bergamaschi et al. in  Keymantic~\cite{Bergamaschi@SIGMOD11Keymantic}, KEYRY~\cite{Bergamaschi@ER2011KEYRY}, and QUEST~\cite{Bergamaschi@VLDBDEMO13_QUEST}. All these systems can be classified as schema-based since they aim at generating a suitable SQL query given an input keyword query. They do not, however, rely on the concept of CJNs, as {\metodo} and all DISCOVER-based systems do. 
Keymantic~\cite{Bergamaschi@SIGMOD11Keymantic} and KEYRY~\cite{Bergamaschi@ER2011KEYRY} consider a scenario in which data instances are not acessible, such as in databases on the hidden web and sources hidden behind wrappers in data integration settings, where typically only metadata is made available. Both systems rely on similarity techniques based on structural and lexical knowledge that can be extracted from the available metadata, e.g., names of attributes and tables, attribute domains, regular expressions, or from other external sources, such as ontologies, vocabularies, domain terminologies, etc. The two systems mainly differ in the way they rank the possible interpretations they generate for an input query. While Keymantic relies on an extension the authors proposed for the Hungarian algorithm, KEYRY is based on the Hidden Markov Model, a probabilistic sequence model, adapted for keyword query modeling.   QUEST~\cite{Bergamaschi@VLDBDEMO13_QUEST} can be thought of as an extension of KEYRY because it uses a similar strategy to rank the mappings from keywords to database elements. QUEST, on the other hand, considers the database instance to be accessible and includes features derived from it for ranking interpretations, in contrast to KEYRY. 

From these systems, QUEST is the one most similar to {\metodo}. However, it is difficult to draw a direct comparison between the two systems as QUEST does not rely on the formal framework from CJN-related previous work \cite{Hristidis@VLDB02DISCOVER,Hristidis@VLDB03Efficient,Oliveira@ICDE15CNRank,  Oliveira@ICDE18MatCNGen,Oliveira@TKDE20} and it also resolves a smaller set of keyword queries then Lathe. 
QUEST, in particular, does not support keyword queries whose resolution necessitates SQL queries with self-joins. As a result, when comparing QUEST to other approaches, the authors limited the experimentation to 35 queries rather then the 50 included in the original benchmark~\cite{Bergamaschi@VLDBDEMO13_QUEST,Coffman@CIKM10Framework}. {\metodo}, on the other hand, supports all 50 queries.

Finally, there are systems that propose going beyond the retrieval of tuples that fulfill a query expressed using
keywords and try to provide a functionality close to structured query languages. This is the case of
SQAK~\cite{Tata@SIGMOD08SQAK} that allows users to specify aggregation functions over schema elements.
Such an approach was later expanded in systems such as SODA~\cite{Blunschi@VLDB2012_soda} and 
SQUIRREL~\cite{Ramada@InfoSys2020SQUIRREL}, which aim to handle not only aggregation functions, 
but also keywords that represent predicates, groupings, orderings and so on. 
To support such features, these systems rely on a variety of resources that are not part of the database schema or
instances. Among these are conceptual schemas, generic and domain-specific ontologies, lists of reserved keywords, and user-defined metadata patterns. We see such useful systems as being closer to natural language query systems~\cite{Affolter@VLDBJ2019_SurveyNLIDB}.
In contrast, {\metodo}, like any typical R-KwS system, 
aims at retrieving sets of JNTs that fulfill the query, and not computing results with 
the tuples. In addition, it does not rely on any external resources.

%% file: sections/3-overview.tex

In this section we present an overview of {\metodo}. We begin by presenting a simple example of the task carried out by our system. For this, we illustrate in Figure~\ref{fig:imdb_instance} a simplified excerpt from the well-known IMDB\footnote{Internet Movie Database https://www.imdb.com/interfaces/}.

\begin{figure}[htb]
	\centering
	\setlength\fboxsep{10pt} 
	\begin{adjustbox}{max width=\textwidth}			
		\begin{tabular}{lll}
			\input{tables/tab_person}
			&
			\multicolumn{2}{l}{\input{tables/tab_movie}}
			\\
			\\
			\input{tables/tab_character}
			&
			\input{tables/tab_role}
			& 
			\input{tables/tab_cast}
		\end{tabular}
	\end{adjustbox}	
	\caption{A simplified excerpt from IMDB}
	\label{fig:imdb_instance}
\end{figure}


Consider that a user inputs the keyword query $Q{=}{``will\,smith\,films"}$ and assume that she wants the system to list the movies in which Will Smith appears. Notice that, informally, the terms ``\textit{will}'' and ``\textit{smith}'' are likely to match the contents of a relation from the DB, while the term ``films'' is likely to match the name of a relation or attribute. 

As other methods previously proposed in the literature, such as CNGen~\cite{Hristidis@VLDB02DISCOVER} and MatCNGen \cite{Oliveira@ICDE18MatCNGen,Oliveira@TKDE20}, the main goal of  {\metodo} is, given a query such as $Q$, generating a SQL query that,
when executed, fulfills the information needed for the user. The difference between {\metodo} and these previous methods is that they are not able to handle references to schema elements, such as ``films'' in $Q$.

For query $Q$, two of the possible SQL queries 
that would be generated are queries $S_1$ and $S_2$, presented in Figures~\ref{fig:imdb_sample_sql_queries}~(a) and~(b), respectively. 
The respective results of these queries for the database of Figure~\ref{fig:imdb_instance} are presented in Figures~\ref{fig:imdb_sample_sql_queries}(c) and~(d).
Query $S_1$ retrieves the movies which Will Smith was in, and thus, satisfies the original user intent. On the other hand, query $S_2$ retrieves movies in which two different persons, whose names respectively include the terms ``\textit{will}'' and ``\textit{smith}'', participated in.

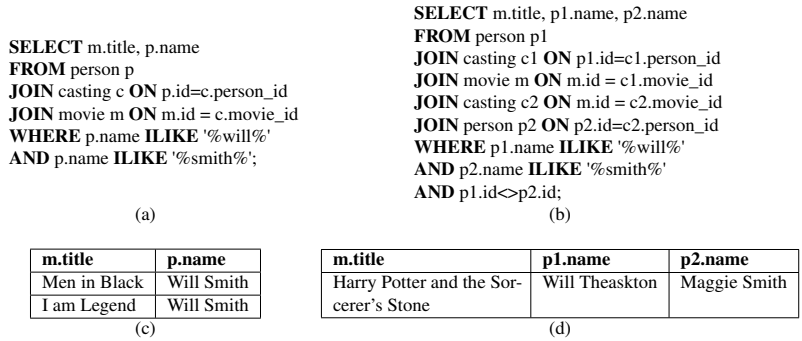
\begin{figure}
	\newsavebox{\mylistingbox}
	
	\begin{lrbox}{\mylistingbox}
		\input{tables/tab_SQL_queries}
	\end{lrbox}
	\scalebox{0.7}{\usebox{\mylistingbox}}
	\caption{SQL queries generated for the keyword query ``\textit{will smith movies}'' and their returned results.}
	\label{fig:imdb_sample_sql_queries}
\end{figure}

	

As this example indicates, there may be several plausible SQL queries related to a given keyword query. Therefore, it is necessary to decide which alternative is more likely to fulfill the user intent. This task is also carried out by \metodo. 

Next, we present an overview of the components and the functioning of \metodo.

\subsection{System Architecture} 
\label{sec:architecture}

In this section, we present the overall architecture of {\metodo}. We base our discussion on Figure~\ref{fig:lathe_arch}, which illustrates the main phases that comprise the operation of the method.

\begin{figure*}[!htb]
	\centering
	\includegraphics[width=\textwidth]{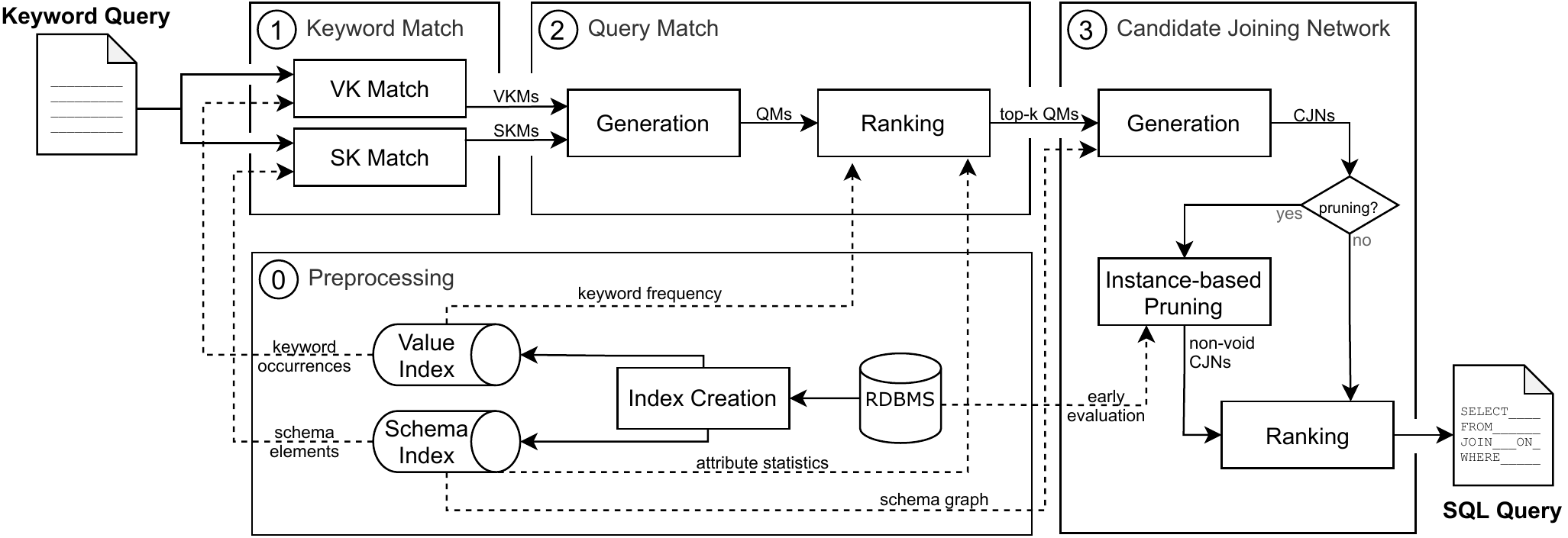}
	\caption{Main phases and architecture of {\metodo}}
	\label{fig:lathe_arch}
\end{figure*}

The process begins with an input keyword query posed by the user. The system then attempts to associate each of the keywords from the query with a database schema element, such as a relation or an attribute.
The system relies on the DB schema, i.e., the names of relations and attributes, 
or on the DB instance, i.e., on the values of the attributes, for this.
This phase, called \textit{Keyword Matching}~\circled{1}, generates sets of  \textit{Value-Keyword Matches} (VKMs), which associate keywords with sets of tuples whose attribute values contain these keywords, and \textit{Schema-Keyword Matches} (SKMs), which associate keywords with names of relations or attributes deemed as similar to these keywords. 

In the next phase, \textit{Query Matching}~\circled{2},
{\metodo} combines keyword matches (VKMs and SKMs) so that they form a \textit{total} and \textit{minimal} cover for the keywords from the query.
That is, the keyword matches comprise all the keywords from the query and no keyword match is redundant. 
These keyword match combinations are referred to as \textit{Query Matches} (QMs). 
Although the Query Matching phase may generate a large number of QMs due to its combinatorial nature, only a few of them are useful in producing plausible answers to the user. 
As a result,  we propose the first algorithm for \textit{Ranking Query Matches} in the literature.
This ranking assigns a score to QMs based on their likelihood of satisfying the needs of the user when formulating the keyword query. 
Thus, the system only outputs a few top-ranked QMs to the next phases. By doing so, it avoids having to process less likely QMs. 

Lastly, in the \textit{Candidate Joining Network Generation}~\circled{3} phase, the system searches for interpretations for the keyword query. That is, the system tries to connect all the keyword matches from the QMs through CJNs, which are based on the schema graph. CJNs can be thought as relational algebra joining expressions that can be directly translated into SQL queries. 
Then, we use the ranked QMs to generate a \textit{Candidate Joining Network Ranking}. 
This ranking favors CJNs that are more concise in terms of the number of relations they employ. 
Once we have identified the most likely CJNs, they can be evaluated as SQL queries that are executed by a DBMS to the users. 
We notice that some of the generated CJNs may return empty results when they are evaluated. Thus, {\metodo} can alternatively evaluate CJNs before ranking them and prune such void CJNs. We call this process \emph{instance-based pruning}.

During the whole process of generating CJNs, {\metodo} uses two data structures which are created in a \textit{Preprocessing stage}~\circled{0}: the \textit{Value Index} and the  \textit{Schema Index}. During the whole process of generating CJNs, {\metodo} uses two data structures which are created in a \textit{Preprocessing stage}~\circled{0}: the \textit{Value Index} and the  \textit{Schema Index}.

The Value Index is an inverted index that stores keyword occurrences in the database, indicating the relations, attributes, and tuples where a keyword appears. 
These occurrences are retrieved to generate VKMs. Furthermore, the Value Index is used to calculate \textit{term frequencies} for the QMs and CJNs Rankings. 
The Schema Index is an inverted index that stores database schema information, as well as statistics about relations and attributes. 
While database schema information, such as PK/FK relationships, are used for the generation of CJNs, the statistics about attributes, such as \textit{norm} and \textit{inverted frequency}, are used for rankings of QMs and CJNs.  

In the following sections we present each of the phases of Figure~\ref{fig:lathe_arch}, describing the steps, definitions, data structures, and algorithms we used.

%% file: tables/tab_person.tex
\begin{tabular}{@{\hspace{0.75em}}r@{}|c|l|}
	\multicolumn{1}{l}{}
	& \multicolumn{2}{@{}l}{\textbf{PERSON}}\\
	\cline{2-3} 
	
	\multicolumn{1}{l|}{}
	& \textbf{ID}
	& \textbf{Name}\\
	\cline{2-3} 
	
	\deftuple{ws}
	& \reftupleno{ws}
	& Will Smith\\
	
	\deftuple{wt}
	& \reftupleno{wt}
	& Will Theakston\\
	
	\deftuple{ms}
	& \reftupleno{ms}
	& Maggie Smith\\
	
	\deftuple{sb}
	& \reftupleno{sb}
	& Sean Bean\\
	
	\deftuple{ew}
	& \reftupleno{ew}
	& Elijah Wood\\
	
	\deftuple{aj}
	& \reftupleno{aj}
	& Angelina Jolie\\
	\cline{2-3} 
	
	\multicolumn{1}{l}{}
	& \multicolumn{1}{l}{}
	& \multicolumn{1}{l}{}
\end{tabular}

%% file: tables/tab_movie.tex
\begin{tabularx}{0.78\textwidth}{r@{}|c|X|l|}
	\multicolumn{1}{l}{}
	& \multicolumn{3}{@{}l}{\textbf{MOVIE}}\\ 
	\cline{2-4}
	
	\multicolumn{1}{l|}{}
	& \textbf{ID}
	& \textbf{Title}
	& \textbf{Year}\\ 
	\cline{2-4}
	
	\deftuple{mib}
	& \reftupleno{mib}
	& Men in Black
	& 1997\\
	
	\deftuple{legend}
	& \reftupleno{legend}
	& I am Legend
	& 2007\\
	
	\deftuple{hp1}
	& \reftupleno{hp1}
	& Harry Potter and the Sorcerer's Stone
	& 2001\\
	
	\deftuple{lotr1}
	& \reftupleno{mib}
	& The Lord of the Rings: The Fellowship of the Ring 
	& 2001\\
	
	\deftuple{lotr2}
	& \reftupleno{lotr2}
	& The Lord of the Rings: The Return of the King
	& 2003\\
	
	
	\deftuple{mrs}
	& \reftupleno{mrs}
	& Mr. \& Mrs. Smith
	& 2005\\
	
	\cline{2-4}
	
	\multicolumn{1}{l}{}
	& \multicolumn{2}{l}{} 
	
\end{tabularx}

%% file: tables/tab_character.tex
\begin{tabular}{r@{}|c|l|}
	\multicolumn{1}{l}{}
	& \multicolumn{2}{@{}l}{\textbf{CHARACTER}}\\
	\cline{2-3} 
	
	\multicolumn{1}{l|}{}
	& \textbf{ID}
	& \textbf{Name}\\
	\cline{2-3} 
	
	\deftuple{agentj}
	& \reftupleno{agentj}
	& Agent J\\
	
	\deftuple{neville}
	& \reftupleno{neville}
	& Robert Neville\\
	
	\deftuple{flint}
	& \reftupleno{flint}
	& Marcus Flint\\
	
	\deftuple{minerva}
	& \reftupleno{minerva}
	& Minerva McGonagall\\
	
	\deftuple{boromir}
	& \reftupleno{boromir}
	& Boromir\\
	
	\deftuple{frodo}
	& \reftupleno{frodo}
	& Frodo Baggins\\
	

    \deftuple{jane}
	& \reftupleno{jane}
	& Jane Smith\\
	
	\cline{2-3}                  
\end{tabular}

%% file: tables/tab_role.tex
\begin{tabular}{r@{}|c|l|}
	\multicolumn{1}{l}{}
	& \multicolumn{2}{@{}l}{\textbf{ROLE}}\\ 
	\cline{2-3}
	
	\multicolumn{1}{l|}{}
	& \textbf{ID}
	& \textbf{Name}\\ 
	\cline{2-3}
	
	\deftuple{actor}
	& \reftupleno{actor}
	& Actor\\
	
	\deftuple{actress}
	& \reftupleno{actress}
	& Actress\\
	
	\deftuple{producer}
	& \reftupleno{producer}
	& Producer\\
	
	\deftuple{writer}
	& \reftupleno{writer}
	& Writer\\
	
	\deftuple{director}
	& \reftupleno{director}
	& Director\\
	
	\deftuple{editor}
	& \reftupleno{editor}
	& Editor\\ 
	
	\cline{2-3}
\end{tabular}

%% file: tables/tab_cast.tex
\begin{tabular}{r@{}|c|c|c|c|c|}
	\multicolumn{1}{l}{}  
	& \multicolumn{5}{@{}l}{\textbf{CASTING} } 
	\\ 
	\cline{2-6}
	
	
	
	\multicolumn{1}{l|}{} 
	& \textbf{ID}
	& \textbf{PID} 
	& \textbf{MID} 
	& \textbf{ChID} 
	& \textbf{RID} \\ 
	
	\cline{2-6}
	
	\deftuple{ws-mib-agentj}
	&\reftupleno{ws-mib-agentj}
	&\reftupleno{ws}
	&\reftupleno{mib}
	&\reftupleno{agentj}
	&\reftupleno{actor}\\
	
	\deftuple{ws-legend-neville}
	&\reftupleno{ws-legend-neville}
	&\reftupleno{ws}
	&\reftupleno{legend}
	&\reftupleno{neville}
	&\reftupleno{actor}\\
	
	\deftuple{wt-hp1-flint}
	&\reftupleno{wt-hp1-flint}
	&\reftupleno{wt}
	&\reftupleno{hp1}
	&\reftupleno{flint}
	&\reftupleno{actor}\\
	
	\deftuple{ms-hp1-minerva}
	&\reftupleno{ms-hp1-minerva}
	&\reftupleno{ms}
	&\reftupleno{hp1}
	&\reftupleno{minerva}
	&\reftupleno{actress}\\
	
	\deftuple{sb-lotr1-boromir}
	&\reftupleno{sb-lotr1-boromir}
	&\reftupleno{sb}
	&\reftupleno{lotr1}
	&\reftupleno{boromir}
	&\reftupleno{actor}\\
	
	\deftuple{sb-lotr2-boromir}
	&\reftupleno{sb-lotr2-boromir}
	&\reftupleno{sb}
	&\reftupleno{lotr2}
	&\reftupleno{boromir}
	&\reftupleno{actor}\\
	
	
	\deftuple{ew-lotr1-frodo}
	&\reftupleno{ew-lotr1-frodo}
	&\reftupleno{ew}
	&\reftupleno{lotr1}
	&\reftupleno{frodo}
	&\reftupleno{actor}\\
	
	\deftuple{ew-lotr2-frodo}
	&\reftupleno{ew-lotr2-frodo}
	&\reftupleno{ew}
	&\reftupleno{lotr2}
	&\reftupleno{frodo}
	&\reftupleno{actor}\\
	
	\deftuple{aj-mrs-jane}
	&\reftupleno{aj-mrs-jane}
	&\reftupleno{aj}
	&\reftupleno{mrs}
	&\reftupleno{jane}
	&\reftupleno{actress}\\
	
	\cline{2-6}
\end{tabular}

%% file: tables/tab_SQL_queries.tex
\begin{tabular}{cc}
	\begin{lstlisting}
	SELECT m.title, p.name
	FROM person p
	JOIN casting c ON  p.id=c.person_id
	JOIN movie m ON m.id = c.movie_id
	WHERE p.name ILIKE '%will%' 
	AND p.name ILIKE '%smith%';
	\end{lstlisting} 
	&
	\begin{lstlisting}
	SELECT m.title, p1.name, p2.name
	FROM person p1
	JOIN casting c1 ON p1.id=c1.person_id
	JOIN movie m ON m.id = c1.movie_id
	JOIN casting c2 ON m.id = c2.movie_id
	JOIN person p2 ON p2.id=c2.person_id
	WHERE p1.name ILIKE '%will%' 
	AND p2.name ILIKE '%smith%'
	AND p1.id<>p2.id;
	\end{lstlisting}
	\\
	(a) & (b) \\\\
	\begin{tabular}{|l|l|} 
		\hline
		\textbf{m.title} & \textbf{p.name}  \\ 
		\hline
		Men in Black     & Will Smith       \\
		\hline
		I am Legend      & Will Smith       \\
		\hline
	\end{tabular}
	&
	\begin{tabular}{|p{3.5cm}|l|l|} 
		\hline
		\textbf{m.title}                      & \textbf{p1.name} & \textbf{p2.name}  \\ 
		\hline
		Harry Potter and the Sorcerer's Stone & Will Theaskton   & Maggie Smith      \\
		\hline
	\end{tabular}
	\\
	(c) & (d) \\
\end{tabular}

%% file: sections/4-keyword-matching.tex

In this section, we present the details on keyword matches and their generation. Their role in our work is
to associate each keyword from the query to some attribute or relation in the database schema. Initially, we classify them as either VKMs and SKMs, according to the type of associations they represent. Later, we provide a generalization of the keyword matches and we introduce the concept of \emph{Keyword-Free Matches}, which will be used in the next phases of our method.

\subsection{Value-Keyword Matching}

We may associate the keywords from the query to some attribute in the database schema-based on the values of this attribute in the tuples that contain these keywords using \emph{value-keyword matches}, according to Definition~\ref{def:value-keyword-match}. 

\begin{definition}
	\label{def:value-keyword-match}	
	Let $Q$ be a keyword query and $R$ be a relation state over the relation schema $R(A_1,\ldots,A_m)$. A \textbf{value-keyword match}  from $R$ over $Q$ is given by:
	\[R^V[A_1^{K_1} , \ldots, A_m^{K_m}]	= \{ t|t \in R
	\wedge \forall A_i: W(t[A_i]) \cap Q = K_i\}  \]	
	where  
	$K_i$ is the set of keywords from $Q$ that are associated to the attribute $A_i$,
	$W(t[A_i])$ returns the set of words in $t$ for attribute $A_i$
	and $V$ denotes a match of keywords to the database values.
\end{definition}
	
Notice that each tuple from the database can be a member of only one value-keyword match. Therefore, the VKMs of a given query are \textbf{disjoint sets of tuples}.

Throughout our discussion, for the sake of compactness in the notation, we often omit mappings of attributes to empty keyword sets in the representation of a VKM. For instance,
we use the notation $R^V[A_1^{K_1}]$ to represent $R^V[A_1^{K_1} , A_2^{ \{\} } , \ldots, A_n^{ \{\} }]$.	

\begin{example}
	\label{ex:valuekeywordmatch}
	Consider the  database instance of Figure~\ref{fig:imdb_instance}. The following VKMs can be generated for the query \textit{``will smith films''}.
	\begin{alignat*}{2}
	&PERSON^V[name^{ \{will,smith\} }]  &= \{\mreftuple{ws}\}\\
	&PERSON^V[name^{ \{will\} }] 		&= \{\mreftuple{wt}\}\\
	&PERSON^V[name^{ \{smith\} }]		&= \{\mreftuple{ms}\}
	\end{alignat*}
\end{example}

VKMs play a similar role to the tuple-sets from related literature~\cite{Hristidis@VLDB02DISCOVER,Oliveira@ICDE18MatCNGen}. 
They are, however, more expressive because they specify which attribute is associated with each keyword. 
Previous R-KwS systems based on the DISCOVER system, on the other hand, are unable to create tuple-sets that span multiple attributes~\cite{Hristidis@VLDB02DISCOVER,Hristidis@VLDB03Efficient,Oliveira@ICDE15CNRank}. Example~\ref{ex:valuekeywordmatch-multipleattributes} shows a keyword query that includes more than one attribute.

\begin{example}
	\label{ex:valuekeywordmatch-multipleattributes}
	Consider the query ``\textit{lord rings 2001}'' whose intent is to return which \textit{Lord of the Rings} movie was launched in 2001. We can represent it with the following value-keyword match:
	\begin{alignat*}{2}
	&MOVIE^V[title^{ \{lord,rings\} }, year^{ \{2001\} }] &= \{\mreftuple{lotr1}\}
	\end{alignat*}
\end{example}

The generation of VKMs uses a structure we call the \textit{Value Index}.  This index stores the 
occurrences of keywords in the database, indicating the relations and tuples a keyword appears and which attributes are mapped to the keyword.
{\metodo} creates the Value Index during a preprocessing phase that scans all target relations only once. 
This phase comes before the query processing and it is not expected to be repeated frequently. 
As a result, without further interaction with the DBMS, answers are generated for each query. The Value Index has following the structure, which is shown in Example~\ref{ex:valueindex}.
\[I_V=\{term:\{relation:\{attribute:\{tuples\}\}\}\}\]

\begin{example}
	The VKMs presented in Example~\ref{ex:valuekeywordmatch} are based on the following keyword occurrences:.
	\begin{alignat*}{3}
	&I_V[will]&=&\{PERSON:\{name:\{\mreftuple{ws},\mreftuple{wt}\}\}\} \\
	&I_V[smith]&=&\{PERSON:\{name:\{\mreftuple{ws},\mreftuple{ms}\}\}\} \\
	&I_V[smith][PERSON]&=&\{name:\{\mreftuple{ws},\mreftuple{ms}\}\} \\
	&I_V[smith][PERSON][name]&=&\{\mreftuple{ws},\mreftuple{ms}\}
	\end{alignat*}
	\label{ex:valueindex}
\end{example}

In {\metodo}, the generation of VKMs is carried out by the VKMGen algorithm, presented in details in 
\ref{apx:vkmgen}.



\subsection{Schema-Keyword Matching}

We may associate the keywords from the query to some attribute or relation in the database schema based on the name of the attribute or relation using \emph{Schema-Keyword Matches}, according to Definition~\ref{def:schema-keyword-match}.
Specifically, our method matches keywords to the names of relations and attributes using similarity metrics.

\begin{definition}
	\label{def:schema-keyword-match}	
	
	Let $k \in Q$ be a keyword from the query, $R(A_1,\ldots,A_m)$ be a relation schema.	A \textbf{schema-keyword match}  from $R$ over $k$ is given by:
	\[R^S[A_1^{K_1} , \ldots, A_m^{K_m}]= \{ t|t \in R \wedge \forall k \in K_i: sim(A_i, k) \geq \varepsilon \} \]
	where 
	$1 \leq i \leq m$,
	$K_i$ is the set of keywords from $Q$ that are associated with the schema element $A_i$, 
	$sim(A_i, k)$ gives the similarity between the name of a schema element $A_i$ and the keyword $k$, which must be above a threshold $\varepsilon$,
	and $S$ denotes a match of keywords to the database schema. 
\end{definition}

In this representation, we use the artificial attribute $self$ when we match a keyword to the name of a relation. Example~\ref{ex:schemakeywordmatch} shows an instance of a schema-keyword match wherein the keyword ``$films$'' is matched to the relation $MOVIE$.

\begin{example}
	\label{ex:schemakeywordmatch}
	The following schema-based relation matches are created for the query ``will smith films'', considering a threshold $\varepsilon=0.6$.
	\begin{alignat*}{3}
	&MOVIE^S[self^{ \{films\} }] 	&=& \{\mreftuple{mib},\mreftuple{legend},\mreftuple{hp1},\mreftuple{lotr1},\mreftuple{lotr2},\mreftuple{mrs}\}\\
	&MOVIE^S[title^{ \{will\} }] 	&=& \{\mreftuple{mib},\mreftuple{legend},\mreftuple{hp1},\mreftuple{lotr1},\mreftuple{lotr2},\mreftuple{mrs}\}\\
	&PERSON^S[name^{ \{smith\} }]	&=& \{\mreftuple{ws},\mreftuple{wt},\mreftuple{ms},\mreftuple{sb},\mreftuple{ew}\}
	\end{alignat*}
	where $sim(a,b)$ gives the similarity between the schema element $a$ and the keyword $b$, $sim(movie, films)=1.00$, $sim(title, will)=0.87$ and $sim(name, smith)=0.63$.
\end{example}

Despite their similarity to VKMs, the schema-keyword matches serve a different purpose in our method, ensuring that the attributes of a relation appear in the query results. 
As a result, they do not ``filter'' any of the tuples from the database, implying that they do not represent any selection operation over database relations. 

\subsubsection*{Similarity Metrics}

For the matching of keywords to schema elements, we used two similarity metrics based on the lexical database \textit{WordNet}:
the \emph{Path similarity}~\cite{Miller@98Wordnet,Pedersen@ACL04WordNetSimilarity} and the \emph{Wu-Palmer similarity}~\cite{Wu@ACL94WuPalmer,Pedersen@ACL04WordNetSimilarity}. We introduce the WordNet database and the two similarity metrics below.

\paragraph*{WordNet Database}

WordNet  \cite{Miller@98Wordnet} is a large lexical database that resembles a thesaurus, as it groups words based on their meanings. 
One use of WordNet  is  to  measure similarity between words based on the relatedness of their \emph{senses},  the  many  different  meanings  that  words can have \cite{Keselj@09}. 
As a result, the word ``film'' can refer to a movie, as well as the act of recording or the plastic film. 
Each of these senses have a different relation to the sense of a ``show".
Wordnet represents sense relationships , such as \textit{synonymy}, \textit{hyponymy}, and \textit{hypernymy}, to measure similarity between words. Synonyms are two word senses that share the same meaning. In addition, we say that the sense $c_1$ is a hyponym of the sense $c_2$ if $c_1$ is more specific, denoting a subclass of $c_2$. For instance, \textit{``protagonist''} is a hyponym of \textit{``character''}; \textit{``actor''} is a hyponym of \textit{``person''}, and \textit{``movie''} is a hyponym of \textit{``show''}. The hypernymy is the opposite of hyponymy relation. 
Thus, $c_2$ us a hypernymy of $c_1$.

\paragraph*{Path Similarity}

The Path similarity~\cite{Miller@98Wordnet,Pedersen@ACL04WordNetSimilarity} exploits the structure and content of the WordNet database. The relatedness score is inversely proportional to the number of nodes along the shortest path between the senses of two words. If the two senses are synonyms, the path between them has length 1. The relatedness score is calculated as follows:
\[sim_{path}(w_1,w_2)=\max_{\substack{c_1 \in senses(w_1)\\c_2 \in senses(w_2)}}\left[\frac{1}{|shortest\_path(c_1,c_2)|}\right]\]

\paragraph*{Wu-Palmer Similarity}

The Wu-Palmer measure (WUP)~\cite{Wu@ACL94WuPalmer,Pedersen@ACL04WordNetSimilarity}  calculates relatedness by considering the depths of the two synsets $c_1$ and $c_2$ in the WordNet taxonomies, along with the depth of the \textit{Least Common Subsumer}(LCS). 
The most specific synset $c_{3}$ is the LCS, which is the ancestor of both synsets $c_1$ and $c_2$.
Because the depth of the LCS is never zero, the score can never be zero (the depth of the root of a taxonomy is one). Also, the score is 1 if the two input synsets are the same. The WUP similarity for two words $w_1$ and $w_2$ is given by:

\[sim_{wup}(w_1,w_2)=\max_{\substack{c_1 \in senses(w_1)\\c_2 \in senses(w_2)}}\left[2\times\frac{depth(lcs(c_1,c_2))}{depth(c_1,c_2)}\right]\]

As in the case of VKMs, we detail the SKMGen algorithm used in {\metodo} in 
\ref{apx:skmgen}.



\subsection{Generalization of Keyword Matches}

Initially, we presented Definitions~\ref{def:value-keyword-match} and \ref{def:schema-keyword-match} which, respectively, introduce VKMs and SKMs. We chose to explain the specificity of these concepts separately for didactic purposes. 
They are, however, both components of a broader concept, \textit{Keyword Match} (KM), which we define in Definition~\ref{def:keyword-match}. 
In the following phases, this generalization will be useful when merging VKMs and 
SKMs.

\begin{definition}
	\label{def:keyword-match}	
	Let $Q$ be a keyword query and $R$ be a relation state over the relation schema $R(A_1,\ldots,A_m)$. 
	Let $V\!K\!M=R^V[A_1^{K_1^S}, \ldots, A_m^{K_m^S}]$ be a value-keyword match from R over Q. 
	Let $S\!K\!M=R^S[A_1^{K_1^S}, \ldots, A_m^{K_m^S}]$ be a schema-keyword match from R over Q.
	A general \textbf{keyword match}  from $R$ over $Q$ is given by:
    \[
		R^S[A_1^{K_1^S} , \ldots, A_m^{K_m^S}]^V[A_1^{K_1^V} , \ldots, A_m^{K_m^V}]	= 
		V\!K\!M \cap S\!K\!M
    \]
\end{definition}
The representations of VKMs and SKMs in the general notation are given as follows:
\[R^S[A_1^{K_1} , \ldots, A_m^{K_m}]	 = R^S[A_1^{K_1} , \ldots, A_m^{K_m}]^V[A_1^{ \{\} } , \ldots, A_m^{\{\}}]	\]
\[R^V[A_1^{K_1} , \ldots, A_m^{K_m}]	 = R^S[A_1^{ \{\} } , \ldots, A_m^{\{\}}]^V[A_1^{K_1} , \ldots, A_m^{K_m}]	\]

Another concept required for the generation of QMs and CNs is \textit{keyword-free matches}, which we describe in Definition~\ref{def:free-keyword-match}. They are KMs that are not associated with any keyword but are used as auxiliary structures, such as intermediate nodes in CJNs.

\begin{definition}
	\label{def:free-keyword-match}
	We say that a keyword match $K\!M$ given by:
	\[K\!M =R^S[A_1^{K_1^S} , \ldots, A_m^{K_m^S}]^V[A_1^{K_1^V} , \ldots, A_m^{K_m^V}]\]
	is a \textbf{keyword-free match} if, and only if, 
	$\nexists K_i^S {\neq} \{\} 
	\land \nexists K_i^V {\neq} \{\}$,
	where $1 \leq i \leq m$.
\end{definition}

For the sake of simplifying the notation, we will represent a keyword-free match as $R^S[\ ]^V[\ ]$ or simply by $R$.

%% file: sections/5-query-matching.tex

In this section, we describe the processes of generating and ranking QMs, which are combinations of the keyword matches generated in the previous phases that comprise every keyword from the keyword query.

\subsection{Query Matches Generation}

We combine the associations present in the KMs to form total and non-redundant answers for the user. 
In other words, {\metodo} looks for KM combinations that satisfy two conditions: (i) every keyword from the query must appear in at least one of the KMs and (ii) if any KM is removed from the combination, 
the combination no longer meets the first condition. These combinations, called \emph{Query Matches} (QMs), are described in Definition~\ref{def:querymatch}

\begin{definition}
	Let $Q$ be a keyword query. 
	Let $M=\{K\!M_{1}, \ldots, K\!M_{n}\}$ be a set of keyword matches for Q in a certain database instance $I$, where:
	\begin{alignat*}{2}
    K\!M_{i}=&&&R_i^{S}[A_{i,1}^{K^S_{i,1}},\ldots,A_{i,m_i}^{K^S_{i,m_i}}]^{V}[A_{i,1}^{K^V_{i,1}},\ldots,A_{i,m_i}^{K^V_{i,m_i}}]
    \end{alignat*}
	Also,	let 
	$C_{K\!M_i}{=}
	\bigcup_{\substack{
			{1 \leq j \leq m_i} \\
			X \in \{S,V\}
	}}
	K_{i,j}^X$
	and 
	$C_{M}{=}
	\bigcup_{\substack{
			{1 \leq i \leq n\ } \\
	}}
	C_{K\!M_i}$
	be the sets of all keywords associated with $K\!M_i$ and with $M$, respectively.
	We say that $M$ is a \textbf{query match} for $Q$ if, and only if,  $C_{M}$ forms a \textbf{minimal set cover}
	of the keywords in $Q$. That is, $C_{M}= Q$ and $C_{M}{\setminus}C_{K\!M_i} \neq Q$, $\forall K\!M_i \in M$.
	\label{def:querymatch}	
\end{definition}

Notice that a QM cannot contain any keyword-free match, as it would not be minimal anymore. Example~\ref{ex:querymatch} presents combinations of KMs which are or are not QMs.

\begin{example}
	\label{ex:querymatch}
	Considering the KMs from the Examples \ref{ex:valuekeywordmatch} and \ref{ex:schemakeywordmatch}, only some of the following sets are considered QMs for the query ``will smith films'':
	\begin{alignat*}{1}
	&M_1 = \{PERSON^V[name^{ \{will,smith\} }], MOVIE^S[self^{ \{films\} }]\}\\
	&M_2 = \{PERSON^V[name^{ \{will\} }],PERSON^V[name^{ \{smith\} }] , MOVIE^S[self^{ \{films\} }]\}\\
	&M_3 = \{PERSON^V[name^{ \{will\} }], PERSON^V[name^{ \{smith\} }]\}\\
	&M_4 = \{
	PERSON^V[name^{ \{will,smith\} }],
	MOVIE^S[self^{ \{films\} }],
	CHARACTER\}\\
	&M_5 = \{PERSON^V[name^{ \{will,smith\} }], PERSON^V[name^{ \{smith\} }], MOVIE^S[self^{ \{films\} }]\}
	\end{alignat*}
	The sets $M_1$ and $M_2$ are considered QMs. 
	In contrast, the sets of keyword matches $M_3$, $M_4$ and $M_5$  are not QMs. While $M_3$ does not include all query keywords,
	$M_4$ and $M_5$ are not minimal, that is, they have unnecessary KMs.
\end{example}

We present the QMGen algorithm for generating QMs in 
\ref{apx:qmgen}.



%% file: sections/5b-qm-ranking.tex

As described in Section~\ref{chap:overview}, {\metodo} performs a ranking of the QMs generated in the previous step. This ranking is necessary because frequently many QMs are generated, yet, only a few of them are useful to produce plausible answers to the user.

{\metodo} estimates the relevance of QMs based on a \emph{Bayesian Belief Network} model for the current state of the underlying database. 
In practice, this model assess two types of relevance when ranking query matches. 
The the TF-IDF model is used to calculate the \emph{value-based score}, which adapts the traditional Vector space model to the context of relational databases, as done in LABRADOR~\cite{Mesquita@IPM07LABRADOR} and CNRank~\cite{Oliveira@ICDE15CNRank}.
The \emph{schema-based score}, on the other hand, is calculated by estimating the similarity 
between keywords and schema elements names. 

In {\metodo}, only the top-k QMs in the ranking are considered in the succeeding phases. By doing so, we avoid generating CJNs that are less likely to properly interpret the keyword query.

\subsection*{Belief Bayesian Network}

We adopt the Bayesian framework proposed by \cite{Ribeiro@1996BBN} and \cite{Cristo@2003BBN}
for modeling distinct IR problems. This framework is simple and allows for the incorporation of features from distinct models into the same representational scheme. 
Other keyword search systems, such as LABRADOR~\cite{Mesquita@IPM07LABRADOR} and CNRank~\cite{Oliveira@ICDE15CNRank}, have also used it.

In our model, we interpret the QMs as documents, which are ranked for the keyword query.
Figure~\ref{fig:bbn-will-smith-fillms} illustrates an example of the adopted Bayesian Network. 
The nodes that represent the keyword query are located at the top of the network, on the Query Side. 
The Database Side, located at the bottom of the network, contains the nodes that represent the QM that will be scored.
The center of the network is present on both sides and is made up of sets of keywords: the set $V$ of all terms present in the values of the database and the set $S$ of all schema element names. 

\begin{figure}[htb]
	\centering
	\includegraphics[width=\textwidth]{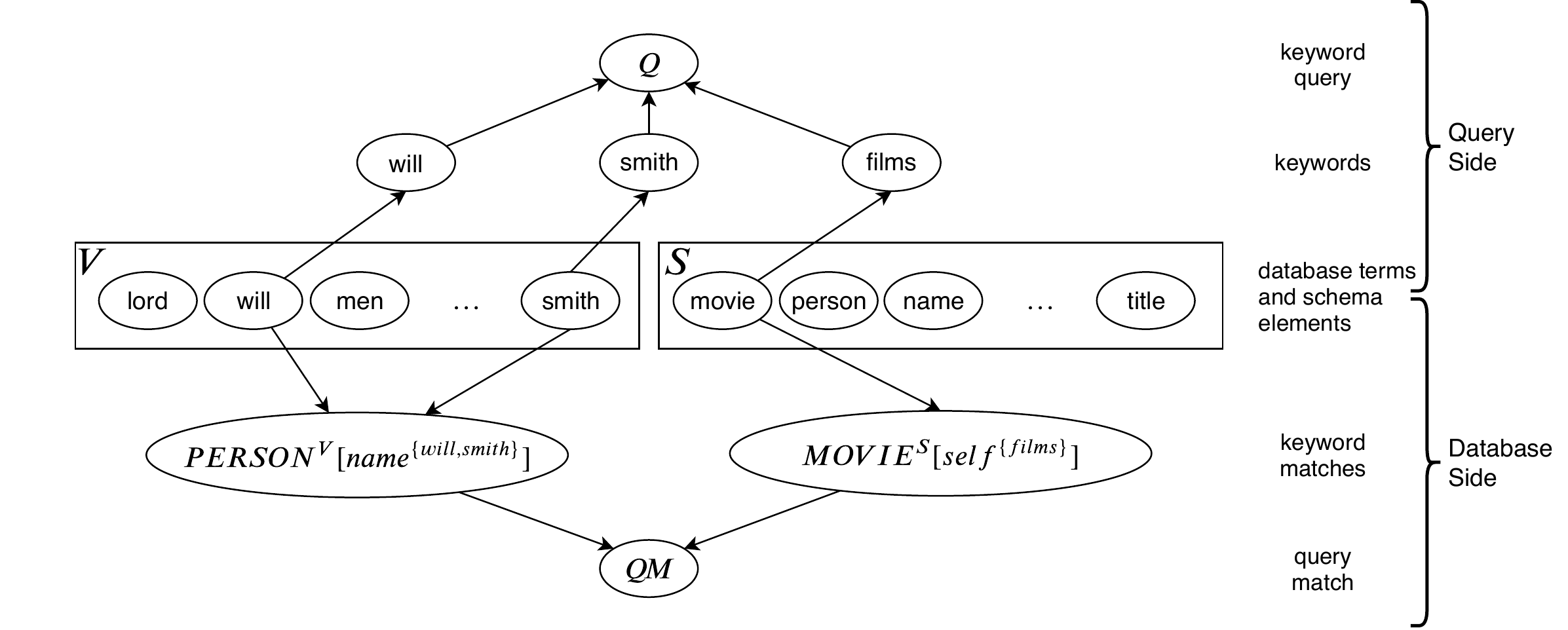}
	\caption{Bayesian Network corresponding to the query $Q=\{will,smith,films\}$}
	\label{fig:bbn-will-smith-fillms}
\end{figure}


In our Bayesian Network, we rank QMs based on their similarities with the keyword query. This similarity is interpreted as the probability of observing a query match $QM$ given the keyword query $Q$, that is, $P(QM|Q)= \mu P(QM \wedge Q)$, where $\mu = 1/P(Q)$ is a normalizing constant, as used in \cite{Pearl@2014Probabilistic}.

Initially, we define a random binary variable associated with each keyword from the sets $V$ and $S$, which indicates whether the keyword was observed in the keyword query. As these random variables are the root nodes of our Bayesian Network, all of the probabilities of the other nodes are dependent on them. Therefore, if we consider $v \subseteq V$ and $s \subseteq S$ as the sets of keywords observed, we can derive the probability of any non-root node $x$ as follows: $P(x)=P(x|v,s)\times P(v)\times P(s)$.

As all the possibilities of $v$ and $s$ are equally likely \emph{a priori}, we can calculate them as $P(v)=(1/2)^{|V|}$ and $P(s)=(1/2)^{|S|}$,  respectively.

The instantiation of the root nodes of the network separates the query match nodes from the query nodes, making them mutually independent. Therefore:
\[P(QM \wedge Q) = P(Q|v,s)P(QM|v,s)P(v)P(s)\]

The probability of the keyword query $Q=\{q_1,\ldots,q_{|Q|}\}$ is split between the probability of each of its keywords:
\[P(Q|v,s) = \prod_{1 \leq i \leq |Q|} P(q_i|v,s)\]
A keyword $q_i$ from the query is observed, given the sets $s$ and $v$, either if $q_i$ occurs in the values of the database or if $q_i$ has a similarity above a threshold $\varepsilon$ with a schema element.
\[P(q_i|v,s) = (q_i \in v) \veebar (\exists k \in s: sim(q_i,k) \geq \varepsilon) \]

Similarly, in our network, the probability of a query match $QM$ is splited between the probability of each of its KMs.
\[P(QM|v,s)= \prod_{1 \leq i \leq |QM|} P(K\!M_{i}|v,s)\]

We compute the probability of KMs using two different metrics: a \emph{schema score} based on the same similarities used in the generation of SKMs; and a \emph{value score} based on a Vector model \cite{Baeza@2011ModernRI,Salton@IPM88TFIDF} using the cosine similarity.

\[P(K\!M_{i}|v,s) =
\prod_{
    \substack{
        1 \leq j \leq m_i \\
        K^{V}_{i,j} \neq \emptyset
    }
} cos(\vv{A_{i,j}},\vv{v \cap K_{i,j}^V})
\prod_{
    \substack{
        1 \leq j \leq m_i \\
        K^{S}_{i,j} \neq \emptyset
    }
}
\frac{
    \sum_{t \in s \cap K_{i,j}^S} sim(A_{i,j},t)
}{|s \cap K_{i,j}^S|}
\]
where $K\!M_{i}=R_i^{S}[A_{i,1}^{K^S_{i,1}},\ldots,A_{i,m_i}^{K^S_{i,m_i}}]^{V}[A_{i,1}^{K^V_{i,1}},\ldots,A_{i,m_i}^{K^V_{i,m_i}}]$.

It is important to distinguish the documents from the Bayesian Network model and the Vector Model. The documents of the Bayesian Network are QMs, and the query is the keyword query itself, whereas the documents of the Vector model are database attributes, and the query is the set of keywords associated with the KM.

Once we know the document and the query of the Vector model, we can calculate the cosine similarity by taking inner product of the document and the query. The cosine similarity formula is given as follows:
\[cos(\vv{A_{i,j}},\vv{v \cap K_{i,j}^V}) = \frac{\vv{A_{i,j}^V} \boldsymbol{\cdot} \vv{v \cap K_{i,j}^V}}{|\vv{A_{i,j}}|\times|\vv{v \cap K_{i,j}^V}|} = \alpha \times
\frac{\displaystyle \sum_{t \in V} w(\vv{A_{i,j}},t) \times w(\vv{v \cap K_{i,j}^V},t)}
{\displaystyle \sqrt{\sum_{t \in V} w(\vv{A_{i,j}},t)^2}}\]
where $\alpha=1 / (\sum_{t \in V}w(\vv{v \cap K_{i,j}^V},t)^2)^{\sfrac{1}{2}}$ is the constant that represents the norm of the query, which is not necessary for the ranking.

The weights for each term are calculated using the TF-IDF measure. This measure is based on the term frequency and specificity in the collection. We use the \emph{raw frequency} and \emph{inverse frequency}, which are the most recommended form of TF-IDF weights \cite{Baeza@2011ModernRI}.

\[
w(\vv{X},t) = freq_{X,t} \times \log{\frac{N_A}{n_t}}
\]
where $\vv{X}\in\{\vv{A_{i,j}},\vv{v \cap K_{i,j}^V}\}$ can be either the document or the query, $N_A$ is the number of attributes in the database, and $n_t$ is the number of attributes that are mapped to the occurrences of the term $t$. In the case of $\vv{X}$ be the query, $freq_{X,t}$ gives the number of occurrences of a term $t$ in the keyword query, which is generally 1. In the case of $\vv{X}$ be an attribute(document),  $freq_{X,t}$ gives the occurrences of a term $t$ in an attribute, which is obtained from the Value Index.

We present the algorithm for ranking QMs in
\ref{apx:qmrank}.





%% file: sections/6-candidate-network-generation.tex

In this section we present the details on our method for generating and ranking Candidate Joining Networks (CJNs), which represent different interpretations of the the keyword query. We recall that our definition of CJNs expands on the definition presented in \cite{Hristidis@VLDB02DISCOVER} to support keywords referring to schema elements.

The generation of CJNs uses a structure we call a \emph{Schema Graph}. In this graph, there is a node representing each relation in the database and the edges correspond to the \emph{referential integrity constraints} (RIC) in the database schema. In practice, this graph is built in a preprocessing phase based on information
gathered from the database schema.

\begin{definition}
	\label{def:schemagraph}
	Let $\mathcal{R} = \{R_1,\ldots,R_n\}$ be a set of relation schemas from the database. Let $E$ be a subset of the ordered pairs from $\mathcal{R}^2$ given by:
	\[E=\{ \langle R_a,R_b \rangle| \langle R_a,R_b \rangle \in \mathcal{R}^2 \wedge R_a \neq R_b \wedge R\!I\!C(R_a,R_b) \geq 1 \}\]
	where $R\!I\!C(R_a,R_b)$ gives the number of \textit{Referential Integrity Constraints} from a relation $R_a$ to
	a relation $R_b$. We say that a \textit{\textbf{schema graph}} is an ordered pair $G_S=\langle \mathcal{R},E \rangle$, where $\mathcal{R}$ is the set of vertices (nodes) of $G_S$, and $E$ is the set of edges of $G_S$.
\end{definition}

\begin{example}
	\label{ex:schemagraph}
	Considering the sample movie database introduced in Figure~\ref{fig:imdb_instance}, our method generates the schema graph bellow.
	\begin{alignat*}{2}
	&G_S=<\ &\{&PERSON,MOVIE,CASTING,CHARACTER,ROLE\},\\
	&&\{&\langle CASTING,PERSON\rangle ,\langle CASTING,MOVIE\rangle ,\\
	&&&\langle CASTING,CHARACTER\rangle ,\langle CASTING,ROLE\rangle \}>
	\end{alignat*}
	
	In Figure~\ref{fig:sgmovie}\add{,} we represent a graphical illustration of $G_S$.
	
	\begin{figure}[htb]
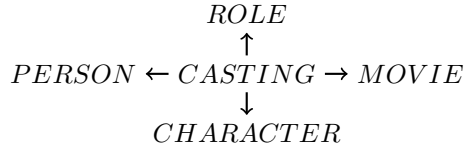

		\centering
		\tikz{
			\node[main node] (1) {$PERSON$};
			\node[main node] (2) [right = 3mm of 1] {$CASTING$};
			\node[main node] (3) [right = 3mm of 2] {$MOVIE$};
			\node[main node] (4) [below = 3mm of 2] {$CHARACTER$};
			\node[main node] (5) [above = 3mm of 2] {$ROLE$};
			\path[draw,thick,->]
			(2) edge node {} (1)
			(2) edge node {} (3)
			(2) edge node {} (4)
			(2) edge node {} (5)
		}
		\caption{A schema graph for the sample movie database of Figure~\ref{fig:imdb_instance}}
		\label{fig:sgmovie}
	\end{figure}
	
\end{example}

Once we defined the schema graph, we can introduce an important concept, the \emph{Joining Network of Keyword Matches} (JNKM). Intuitively, a joining network of keyword matches $J$ contains every KM from a query match $M$. $J$ may also contain some free-keyword matches for the sake of connectivity. Finally, $J$ is a connected graph that is structured according to the schema graph $G_S$. The definition of candidate joining networks is given as follows:

\begin{definition}
	\label{def:candidatenetwork}
	Let $M$ be a query match for a keyword query $Q$. Let $G_S$ be a schema graph. Let $F$ be a set of keyword-free matches from the relations of $G_S$. Consider a graph of keyword matches $J = \langle \mathcal{V}, E\rangle$, where $\mathcal{V}$ and $E$ are the vertices and edges of $J$. We say that $J$ is a \textbf{joining network of keyword matches} from $M$ over $G_S$ if the following conditions hold:
	\begin{alignat*}{2}
	&i) && \mathcal{V} = M  \cup F\\
	&ii)&& \forall K\!M_i \in \mathcal{V} : \exists \langle  K\!M_a,K\!M_b \rangle \in  E | \ i \in \{a,b\} \\
	&iii)&&  \forall \langle K\!M_a,K\!M_b \rangle \in E \implies \exists \langle R_a, R_b \rangle \in G_S\\
	\end{alignat*}
\end{definition}

For the sake of simplifying the notation, we will use a graphical illustration to represent CJNs, which is shown in Example~\ref{ex:non-minimal-candidate-networks}. 

\begin{example}
	\label{ex:non-minimal-candidate-networks}
	Considering the query match $M_1$ previously generated in Example~\ref{ex:querymatch}, the following JNKMs can be generated:
	\[
	\tikz{
		\node[main node] (1) {$J_1=  PERSON^V[name^{ \{will,smith\} }]$};
		\node[main node] (2) [right = 3mm of 1] {$CASTING$};
		\node[main node] (3) [right = 3mm of 2] {$MOVIE^S[self^{\{films\}}]$};
		\path[draw,thick,->]
		(2) edge node {} (1)
		(2) edge node {} (3)
	}
	\]
	\[
	\tikz{
		\node[main node] (1) {$J_2=  PERSON^V[name^{ \{will,smith\} }]$};
		\node[main node] (2) [right = 3mm of 1] {$CASTING$};
		\node[main node] (3) [right = 3mm of 2] {$MOVIE^S[self^{\{films\}}]$};
		\node[main node] (4) [below = 3mm of 2] {$CHARACTER$};
		\path[draw,thick,->]
		(2) edge node {} (1)
		(2) edge node {} (3)
		(2) edge node {} (4)
	}
	\]
	
	The JNKMs $J_1$ and $J_2$ cover the query match $M_1$.
	The interpretation of $J_1$ looks for the \emph{movies} of the \emph{person} \textit{will smith}.  $J_2$ looks for the \emph{movies} of the \emph{person} \textit{will smith} and which \emph{character} will smith played in these movies.
\end{example}

Notice that a JNKM might have unnecessary information for the keyword query, which was the case of $J_2$ presented in Example~\ref{ex:non-minimal-candidate-networks}. One approach to avoid generating unnecessary information is to generate Minimal Joining Networks of Keyword Matches(MJNKM), which are addressed in Definition~\ref{def:minimalcandidatenetwork}. Roughly, a MJNKM cannot have any keyword-free match as a leaf, that is, a keyword-free match incident to a single edge.

\begin{definition}
	\label{def:minimalcandidatenetwork}
	Let $G_S$ be a schema graph. Let $M$ be a query match for a query Q.
	We say that $J=\langle \mathcal{V}, E \rangle$ from $M$ over $G_S$ is \textbf{minimal joining network of keyword matches} (MJNKM) if, and only if, the following condition holds:
	\[ \forall K\!M_i \in \mathcal{V}\ (\ \exists! \langle K\!M_a,K\!M_b \rangle \in E | i \in \{a,b\} \implies K\!M_i \neq R_i^S[\ ]^V[\ ]\ ) \]
\end{definition}

\begin{example}
	\label{ex:minimal-candidate-networks}
	Considering the query match $M_2$ previously generated in Example~\ref{ex:querymatch}, the following MJNKMs can be generated:
	\[
	\tikz{
		\node[main node] (1) {$J_3= PERSON^V[name^{ \{smith\} }]$};
		\node[main node] (2) [right = 3mm of 1] {$CASTING$};
		\node[main node] (3) [right = 3mm of 2] {$PERSON^V[name^{ \{will\} }]$};
		\node[main node] (4) [below = 3mm of 2] {$MOVIE^S[self^{\{films\}}]$};
		\path[draw,thick,->]
		(2) edge node {} (1)
		(2) edge node {} (3)
		(2) edge node {} (4)
	}
	\]
		
\end{example}

Another issue that a JNKM might have is representing an inconsistent interpretation. For instance, it is impossible for $J_3$ presented in Example~\ref{ex:minimal-candidate-networks} to return any results from the database.
By Definition~\ref{def:value-keyword-match}, the VKMs $PERSON^V[name^{ \{will\} }]$ and $PERSON^V[name^{ \{smith\} }]$ are disjoint. However, a tuple from $CASTING$ cannot refer to two different tuples of $PERSON$. Thus $J_3$ is inconsistent.
We notice that previous work in literature for CJN generation had addressed this kind of inconsistency  \cite{Hristidis@VLDB02DISCOVER,Oliveira@ICDE18MatCNGen}. 
They did not, however, consider the situation in which there exist more than one RIC from one relation to another. In contrast, based on the theorems and definitions presented in \cite{Hristidis@VLDB02DISCOVER}, {\metodo} proposes a novel approach for checking consistency in CJNs that support such scenarios. 
Theorem~\ref{thm:sound-cjn} presents a criterion that determines when a JNKM is \emph{sound}, that is, it can only produce JNTs that do not have more than one occurrences of a tuple. The proof of Theorem~\ref{thm:sound-cjn} is presented in
\ref{apx:sound-thm}.

\begin{restatable}{theorem}{soundthm}
    \label{thm:sound-cjn}
    Let $G_S = \langle \mathcal{R}, E_G \rangle$ be a schema graph.
    Let $J =\langle \mathcal{V}, E_J \rangle$ be a joining network of keyword matches.
    We say that $J$ is \textbf{sound}, that is, it does not have more than one occurrences of the same tuple for every instance of the database if, and only if, the following condition holds
    $\forall K\!M_a \in~\mathcal{V}, \forall \langle R_a, R_b \rangle \in E_{G}:$
    \[
        R\!I\!C(R_a,R_b) \geq |\{KM_c| 
        \langle K\!M_a, K\!M_c \rangle 
        \in E_{J} \wedge R_c=R_b \}|
    \]
    where $R\!I\!C(R_a,R_b)$ indicates the number of \textit{Referential Integrity Constraints} from a relation $R_a$ to
	a relation $R_b$.
\end{restatable}

Example~\ref{ex:rics-joining} presents a JNKM that is sound, although it would be deemed not sound by previous approaches \cite{Hristidis@VLDB02DISCOVER,Oliveira@ICDE18MatCNGen}.

\begin{example}
	\label{ex:rics-joining}
	Consider a simplified excerpt from the MONDIAL database \textnormal{\cite{May@99MONDIAL}},
	presented in Figure~\ref{fig:mondial_instance}. As there exists 2 RICs from the relation
	$BORDER$ to $COUNTRY$, represented by the attributes \emph{Ctry1\_Code} e \emph{Ctry2\_Code}, a tuple from $BORDER$ can be joined to at most two distinct tuples from $Country$, which is the case of $\mreftuple{colombia}\bowtie\mreftuple{co_br}\bowtie\mreftuple{brazil}$. Thus, the following MJNKM is sound:
	\[
	\tikz{
		\node[main node] (1) {$J_4=  COUNTRY^V[name^{\{colombia\}}]$};
		\node[main node] (2) [right = 3mm of 1] {$BORDER$};
		\node[main node] (3) [right = 3mm of 2] {$COUNTRY^V[name^{ \{brazil\} }]$};
		\path[draw,thick,->]
		(2) edge node {} (1)
		(2) edge node {} (3)
	}
	\]
	\begin{figure}[!htb]
    	\centering
    	\begin{adjustbox}{max width=\textwidth}			
    		\begin{tabular}{lll}
    			\input{tables/tab_country}
    			&
    			\input{tables/tab_border}
    			&
    			\input{tables/tab_city}
    		\end{tabular}
    	\end{adjustbox}	
    	\caption{A simplified excerpt from MONDIAL}
    	\label{fig:mondial_instance}
    \end{figure}
\end{example}

Finally, Definition~\ref{def:sound-candidate-network} describes a Candidate Joining Network (CJN), which is roughly a sound minimal joining network of keyword matches.

\begin{definition}
	\label{def:sound-candidate-network}
	Let $M$ be a query match for the keyword query $Q$. Let $G_S$ be a schema graph. Let $C\!J\!N$ be a joining network of keyword matches from $M$ over $G_S$ given by $C\!J\!N=\langle \mathcal{V}, E \rangle$.
	We say that $C\!J\!N$ is a \textbf{candidate joining network} if, and only if, $C\!J\!N$ is minimal and sound.
\end{definition}

\begin{example}
	\label{ex:sound-candidate-networks}
	Considering the query match $M_2$ previously generated in Example~\ref{ex:querymatch}, the following CJN can be generated:
	\[
	\tikz{
		\node[main node] (1) {$C\!J\!N_1=  MOVIE^S[self^{\{films\}}]$};
		\node[main node] (2) [right = 3mm of 1] {$CASTING$};
		\node[main node] (3) [right = 3mm of 2] {$PERSON^V[name^{ \{will\} }]$};
		\node[main node] (4) [below = 3mm of 1] {$CASTING$};
		\node[main node] (5) [right = 3mm of 4] {$PERSON^V[name^{ \{smith\} }]$};
		\path[draw,thick,->]
		(2) edge node {} (1)
		(2) edge node {} (3)
		(4) edge node {} (1)
		(4) edge node {} (5)
	}
	\]
	The candidate joining networks $C\!J\!N_1$ covers the query match $M_2$. $C\!J\!N_1$ is a minimal and sound JNKM. The interpretation of $C\!J\!N_1$ searches for the \textit{movies} where both \textit{persons} ``will'' (e.g. Will Theakston) and ``smith'' (e.g. Maggie Smith) participate in. The two keyword-free matches from the $CASTING$ are treated as different nodes in the candidate joining network $C\!J\!N_1$.
\end{example}

The details on how we generate CJNs in {\metodo} are described by the CNKMGen Algorithm in 
\ref{apx:cnkmgen}.



\subsection{Candidate Joining Network Ranking} \label{sec:cn-ranking}

In this section, we present CJNKMRank, a novel ranking of CJNs based on the ranking of QMs. This ranking is necessary because often many CJNs are generated, yet, only a few of them are indeed useful to produce relevant answers.

We present in Section~\ref{sec:query-match-ranking} a QM ranking that advances the majority of the features present in the ranking of CJNs of other proposed systems, such as CNRank~\cite{Oliveira@ICDE15CNRank}. Thus, we can exploit the scores of the QMs to rank the CJNs. For this reason, CJNKMRank provides a straightforward yet effective ranking of candidate joining networks.

CJNKMRank is described in Algorithm~\ref{alg:cnkmrank}. Roughly, it uses the ranking of QMs adding a penalization for large CJNs. Therefore, the score of a candidate joining network $C\!J\!N_M$ from a query match $M$ is given by:
\[
score(C\!J\!N_M) = score(M) \times  \frac{1}{|C\!J\!N_M|}
\]

To ensure that CJNs with the same score are placed in the same order that they were generated, in Line~\ref{line:sorting}, we used a stable sorting algorithm \cite{Cormen@09Algorithms}.

\input{algorithms/alg-CNKMRank}

\subsection{Candidate Joining Network Pruning} \label{sec:cn-pruning}
In this section we present an \emph{eager evaluation} strategy for pruning CJNs. 
Even if CJNs contain valid interpretations of the keyword query, some of them may fail to produce any JNTs as a result. Thus, we can improve the results of our CJN generation and ranking if by pruning what we call \emph{void} CJNs, which are CJNs with no JNTs in their results.

\begin{example}
	\label{ex:empty-candidate-networks}
	Considering the database instance of Figure~\ref{fig:imdb_instance} and the keyword query ``will smith films'', the following CJNs can be generated:
	\[
	\tikz{
	    \node[main node] (2) {$C\!J\!N_2 = MOVIE^S[self^{\{films\}}]^V[name^{ \{smith\} }]$};
	    \node[main node] (3) [right = 3mm of 2] {$CASTING$};
		\node[main node] (4) [below = 3mm of 3]	{$PERSON^V[name^{ \{will\}}]$};
		\path[draw,thick,->]
		(3) edge node {} (2)
		(3) edge node {} (4)
		
	}
	\]
	\[
	\tikz{
		\node[main node] (1) {$C\!J\!N_3=  MOVIE^S[self^{\{films\}}]$};
		\node[main node] (2) [right = 3mm of 1] {$CASTING$};
		\node[main node] (3) [right = 3mm of 2] {$PERSON^V[name^{ \{will\} }]$};
		\node[main node] (4) [below = 3mm of 1] {$CASTING$};
		\node[main node] (5) [right = 3mm of 4] {$CHARACTER^V[name^{ \{smith\} }]$};
		\path[draw,thick,->]
		(2) edge node {} (1)
		(2) edge node {} (3)
		(4) edge node {} (1)
		(4) edge node {} (5)
	}
	\]
		
	
	The interpretation of $C\!J\!N_2$ looks for the \textit{movies} whose name contains the  keyword ``smith''
	(e.g. ``Mr. \& Mrs. Smith'') and in which a person whose contains ``will'' (e.g. ``Will Theakston'') participate in.
	The interpretation of $C\!J\!N_3$ looks for the \textit{movies} where a person whose name contains ``will'' (e.g. ``Will Theakston'') played the character ``smith'' (e.g. ``Jane Smith'').
	Notice that although the candidate joining networks $C\!J\!N_2$ and $C\!J\!N_3$ both provide valid interpretations for the keyword query, they do not produce any tuples as a result in the given database instance.
	
\end{example}

As most of the previous work does not rank CJNs but only evaluates them and ranks their resulting JNTs instead, the pruning of void CJNs has previously never been addressed.
{\metodo} employs a pruning strategy that evaluates CJNs as soon as they are generated, pruning the void ones. This strategy, as demonstrated in our experiments, can significantly improve the quality of the CJN generation process, particularly in scenarios where the schema graph contains a large number of nodes and edges.

For instance, one of the datasets we use in our experiments, the MONDIAL database, contains a large number of relations and relational integrity constraint (RICs). 
This results in a schema graph with several nodes and edges, which, intuitively, incur a large number of possible CJNs for a single QM.
In contrast, we discovered that such schema graphs are prone to produce a large number of void CJNs.
In particular, while approximately 20\% of the keyword queries used in our experiments required us to consider 9 CJNs per QM, the eager evaluation strategy reduced this value to 2 CJNs per QM.

Notice, however, that to find if some CJN is void, we must execute it as an SQL in the DBMS, which incurs an additional cost and an increase in the CJN generation time. 
Despite that, we notice in our experiments that the eager evaluation strategy does not necessarily hinder the performance of a R-KwS system.   In fact, the reducing the number of CJNs per QM alone improves the system efficiency because this parameter influences the CJN generation process.
Furthermore, the eager evaluation advances the CJN evaluation, which is already a required step in the majority of R-KwS systems in the related work.
Lastly, we can set a maximum number of CJNs to probe during the eager evaluation, which limits the increase in CJN generation time.

%% file: tables/tab_country.tex
\begin{tabular}{rcll}
    & 
    \multicolumn{3}{l}{
        \textbf{COUNTRY}
    }                                                                            
    \\
    \cline{2-4} 
    \multicolumn{1}{r|}{}
    &
    \multicolumn{1}{c|}{
        \textbf{Code}
    } 
    &
    \multicolumn{1}{l|}{
        \textbf{Name}
    }
    &
    \multicolumn{1}{c|}{
        \textbf{Capital\_ID}
    }
    \\
    \cline{2-4} 
\multicolumn{1}{r|}{\deftuple{colombia}} & \multicolumn{1}{c|}{CO}            & \multicolumn{1}{l|}{Colombia}      & \multicolumn{1}{c|}{1}           \\
\multicolumn{1}{r|}{\deftuple{brazil}}   & \multicolumn{1}{c|}{BR}            & \multicolumn{1}{l|}{Brazil}        & \multicolumn{1}{c|}{2}         \\
\multicolumn{1}{r|}{\deftuple{peru}}     & \multicolumn{1}{c|}{PE}            & \multicolumn{1}{l|}{Peru}          & \multicolumn{1}{c|}{3}             \\ \cline{2-4} 
\end{tabular}

%% file: tables/tab_border.tex
\begin{tabular}{rccc}
                                                        & \multicolumn{3}{l}{\textbf{BORDER}}                                                                                                \\ \cline{2-4} 
\multicolumn{1}{r|}{}                                   & \multicolumn{1}{c|}{\textbf{Ctry1\_Code}} & \multicolumn{1}{c|}{\textbf{Ctry2\_Code}} & \multicolumn{1}{c|}{\textbf{Length}} \\ \cline{2-4} 
\multicolumn{1}{r|}{\deftuple{co_br}} & \multicolumn{1}{c|}{CO}                      & \multicolumn{1}{c|}{BR}                      & \multicolumn{1}{c|}{1643}            \\
\multicolumn{1}{r|}{\deftuple{pe_br}} & \multicolumn{1}{c|}{PE}                      & \multicolumn{1}{c|}{BR}                      & \multicolumn{1}{c|}{1560}            \\ \cline{2-4} 
\\
\end{tabular}

%% file: tables/tab_city.tex
\begin{tabular}{rccc}
                      & \multicolumn{3}{l}{\textbf{CITY}}                                                                                \\ \cline{2-4} 
\multicolumn{1}{r|}{} & \multicolumn{1}{c|}{\textbf{ID}} & \multicolumn{1}{c|}{\textbf{Name}} & \multicolumn{1}{c|}{\textbf{Population}} \\ \cline{2-4} 
\multicolumn{1}{r|}{\deftuple{bogota}} & \multicolumn{1}{c|}{1}          & \multicolumn{1}{c|}{Bogota}        & \multicolumn{1}{c|}{1643}                \\
\multicolumn{1}{r|}{\deftuple{brasilia}} & \multicolumn{1}{c|}{2}          & \multicolumn{1}{c|}{Brasilia}      & \multicolumn{1}{c|}{1560}                \\
\multicolumn{1}{l|}{\deftuple{lima}} & \multicolumn{1}{c|}{3}          & \multicolumn{1}{c|}{Lima}          & \multicolumn{1}{c|}{1560}                \\ \cline{2-4} 
\end{tabular}

%% file: algorithms/alg-CNKMRank.tex
\begin{algorithm}[htb]
	\caption{CNKMRank($Q\!M$)}
	\label{alg:cnkmrank}
	\KwIn{A set of candidate networks $C\!N$
	}
	\KwOut{The set of candidate networks $R\!C\!N$}
	$R\!C\!N \leftarrow [\ ]$\;
	\For{$C \in R\!C\!N$}{
		\Let{$M$ be the query match used to generate $C$}
		
		cn\_score = $score(M) /|C|$
		
		$R\!C\!N\!$.append( $\langle cn\_score, C \rangle $ )
	}
	
	\textbf{Sort }$R\!C\!N$ in descending order\; \label{line:sorting}
	
	\Return{$R\!C\!N$}	
\end{algorithm}

%% file: sections/7-experiments.tex

In this section, we report a set of experiments performed using datasets and query sets previously used in similar experiments reported in the literature. Our goal is to evaluate the quality of the CJN Ranking, the quality QM ranking, and how we can improve the CJN Generation using an Eager Evaluation strategy.

\subsection{Experimental Setup}
\label{subsec:setup}

\paragraph{System Details}
We ran the experiments on a Linux machine running Arch Linux (LSB 1.4, 64-bit, 16GB RAM, Intel\textsuperscript{\textregistered} Core\textsuperscript{\texttrademark} i5-4670 CPU @ 3.40GHz). 
We used PostgreSQL as the underlying RDBMS with a default configuration. All implementations were made in Python 3.


\paragraph{Datasets}

\newcommand{\ftna}{\footnote{https://www.imdb.com/}}
\newcommand{\ftnb}{\footnote{https://www.cia.gov/library/publications/the-world-factbook/}}
\newcommand{\fttd}{\footnote{https://towardsdatascience.com/understanding-boxplots-5e2df7bcbd51}}

For all the experiments, we used two datasets, \emph{IMDb} and \emph{MONDIAL}, which were
used for the experiments performed with previous R-KwS systems and methods 
\cite{Coffman@CIKM10Framework,Coffman@TKDE12Evaluation,Luo@SIGMOD07Spark,Oliveira@ICDE15CNRank,Oliveira@ICDE18MatCNGen,Oliveira@TKDE20}.The IMDb dataset is a subset of the well-known Internet Movie Database (IMDb)\ftna, which comprises information related to films, television shows and home videos – including actors, characters, etc.  The MONDIAL dataset \cite{May@99MONDIAL} comprises geographical and demographic information from the well-known \emph{CIA World Factbook}\ftnb, the International Atlas, the TERRA database, and other web sources. 
The two datasets distinct characteristics. Although the The IMDb dataset is larger, the MONDIAL dataset is more complex, with more relations and relational integrity constraints (RICs). Table~\ref{tab:datasets} summarizes the details of each dataset.

\begin{table}[htb]
	\centering
	\caption{Datasets we used in our experiments} 
	\label{tab:datasets}
	\begin{adjustbox}{max width=\columnwidth}			
		\small
		\input{tables/tab-datasets}	
	\end{adjustbox}	
\end{table}

\paragraph{Query Sets}

\newcommand{\ftqueries}{\footnote{\add{These are the queries 21-25.}} }

We used the query sets provided by Coffman \& Weaver~\cite{Coffman@CIKM10Framework} benchmark for the two datasets. 
However, we notice that several queries from both query sets do not have a clear intent, compromising the proper evaluation of the results, for instance, the ranking of CJNs. Therefore, for the sake of providing a more fair evaluation, we generated an additional for each original query set replacing queries that we consider unclear with equivalent queries with added schema references. 

As an example, consider the query ``\textit{Saint Kitts Cambodia}'' for the MONDIAL dataset, where \textit{Saint Kitts} and \textit{Cambodia} are the names of the two countries. 
There exist several interpretations of this keyword query, each of them with a distinct way to connect the tuples corresponding to these countries.
For example, one might look for shared religions, languages, or ethnic groups between the two countries. 
While all these interpretations are valid in theory, the relevant interpretation defined by Coffman \& Weaver~\cite{Coffman@CIKM10Framework} in their golden standard indicates that the query searches for organizations in which both countries are members. In this case, we replaced in the new query set with the query "Saint Kitts Cambodia Organizations". 

Table~\ref{tab:querysets} presents the query sets we used in our experiments, along with some of their features. Query sets whose names include the suffix ``-DI'' correspond to those in which we have replaced ambiguous queries as explained above. Thus, these queries sets have no ambiguous queries and they have a higher number of Schema References.

\begin{table}[htb]
	\centering
	\caption{Query sets we used in our experiments} 
	\label{tab:querysets}
	\begin{adjustbox}{max width=\textwidth}			
		\small
		\begin{tabular}{llcccc}
			\toprule
			Query Set & Target Dataset & Total Queries & Ambiguous Queries & Schema References \\
			\midrule
			IMDb & IMDb &  50 & 5  & 20 \\	
			IMDb-DI &  IMDb & 50 & -  & 25 \\	
			MOND & MONDIAL & 50 & 7  & 12 \\
			MOND-DI & MONDIAL & 50 & -  & 19 \\
			\bottomrule
		\end{tabular} 
	\end{adjustbox}	
\end{table}

\paragraph{Golden Standards}

The relevant interpretations for each query and its relevant SQL results are provided by Coffman \& Weaver~\cite{Coffman@CIKM10Framework}. We used them to evaluate the quality of our rankings of CJNs and QMs. 
We used the following procedure to obtain the golden standards for CJNs and QMs.
We took the relevant interpretation defined by Coffman \& Weaver~\cite{Coffman@CIKM10Framework} for each keyword query and manually identified the CJN that represents this relevant interpretation to define this CJN as the relevant one. Following that, we took the set of nodes from the relevant CJN that are not free-keyword matches. 
This set of KMs corresponds to the relevant QM.

\paragraph{Metrics}
We evaluate the ranking of CJNs and QMs using two metrics: Precision at position ranking $K$ ($P@K$) and Mean Reciprocal Rank~(MRR).

Given a keyword query $Q$ , the value of $P_Q@K$ is 1 if the target CJN for $Q$ appears in a position up to $K$ in the ranking, and 0 otherwise. $P@K$ is the average of $P_Q@K$, for all $Q$ in a query set.  With respect to MRR, given a keyword query $Q$, the value of the \emph{reciprocal ranking for $Q$}, $RR_Q$ is given by $\frac{1}{K}$, where $K$ is the rank position of the relevant result. Then, the MRR obtained for the queries in a query set is the average of $RR_Q$, for all $Q$ in the query set. Intuitively, the MRR metric measures how close the relevant results are from the first position of the ranking.

\paragraph{{\metodo} Setup} 
For the experiments we report here, we set a maximum size for QMs and CJNs of 3 and 5, respectively. Also, we consider three important parameters for running {\metodo}: $N_{QM}$, the maximum number of QMs considered from the QM ranking; $N_{CJN}$, the maximum number of CJNs considered from each QM; 
and $P_{CJN}$, the number of CJNs probed per QM by the eager evaluation. In this context, a \emph{setup} for {\metodo} is a triple $N_{QM}/N_{CJN}/P_{CJN}$. The most common setup we used in our experiments is
$5/1/9$, in which we take the top-5 QMs in the ranking, generate and probe up to 9 CJNs for each QM, 
and take only the first non-empty CJN, if any, from each QM. We call this the \emph{default setup}.
Later in this section, we will discuss how these parameters affect the effectiveness and the performance 
of {\metodo}, as well as why we use the default configuration.

All the resources, including source code, query sets, datasets and golden standards used in our experiments are available at \url{https://github.com/pr3martins/Lathe}.

\subsection{General Results} \label{sec:gresults}

In this section, we present general results regarding the CJN generation process. 
To have an idea of the full process, to obtain these results we do not apply any pruning strategies, and we took
all the QMs and CJNs generated. Table~\ref{tab:general-results} shows some statistics about to the CJN generation process for all query sets and datasets that we tested. Specifically, we show the maximum and average for the number of KMs, QMs, CJNs. This last number refers to the total number of CJNs generated for all QMs of each query.

\begin{table}[!htb]
    \centering
    \caption{Result Statistics for each Query Set and Dataset.}
    \begin{adjustbox}{max width=0.9\textwidth}	

\input{tables/tab_general_results}
    \end{adjustbox}
    \label{tab:general-results}
\end{table}

Overall, the number of KMs and QMs are higher in the case of IMDb. This is due to IMDb database having a higher number of tuples and the keywords being present in several relations and combinations.
However, despite having a ``dense'' schema graph, the MONDIAL dataset does not have a large number of tuples and the keywords are matched to only a few relations, resulting in a lower number of KMs and QMs generated. 

Regarding the number of CJNs, the results for the IMDb dataset are also higher in comparison with the results from the MONDIAL dataset. However, the ratio of the number CJNs to the number of QMs is higher in the results for the MONDIAL dataset, which is due to its more complex schema graph.

\subsection{Comparison with other R-KwS systems}

\newcommand{\ftquestqueries}{\footnote{Specifically, queries 01-20, 26-35 and 46-50.}}

In this experiment, we first compare {\metodo} with QUEST~\cite{Bergamaschi@VLDBDEMO13_QUEST}, the current state of art R-KwS system with support to schema references and then we also compare {\metodo} with several other R-KwS systems. Here, we used the default {\metodo} setup, that is, $5/1/9$. 
We compare our results to those published by the authors, which refer to the MONDIAL dataset, because we were unable to run QUEST due to the lack of code and enough details for implementing it
Figure~\ref{fig:quest-comparison-35queries} depicts the results for
the 35 queries supported by QUEST\ftquestqueries out of the 50 queries provided in the original query set.
The graphs show the recall and P@1 values 
for the raking produced by each system considering the golden standard supplied by 
Coffman \& Weaver~\cite{Coffman@CIKM10Framework}. 

\begin{figure}[!htb]
    \centering
    \begin{tabular}{rl}
         \includegraphics[width=.5\textwidth]{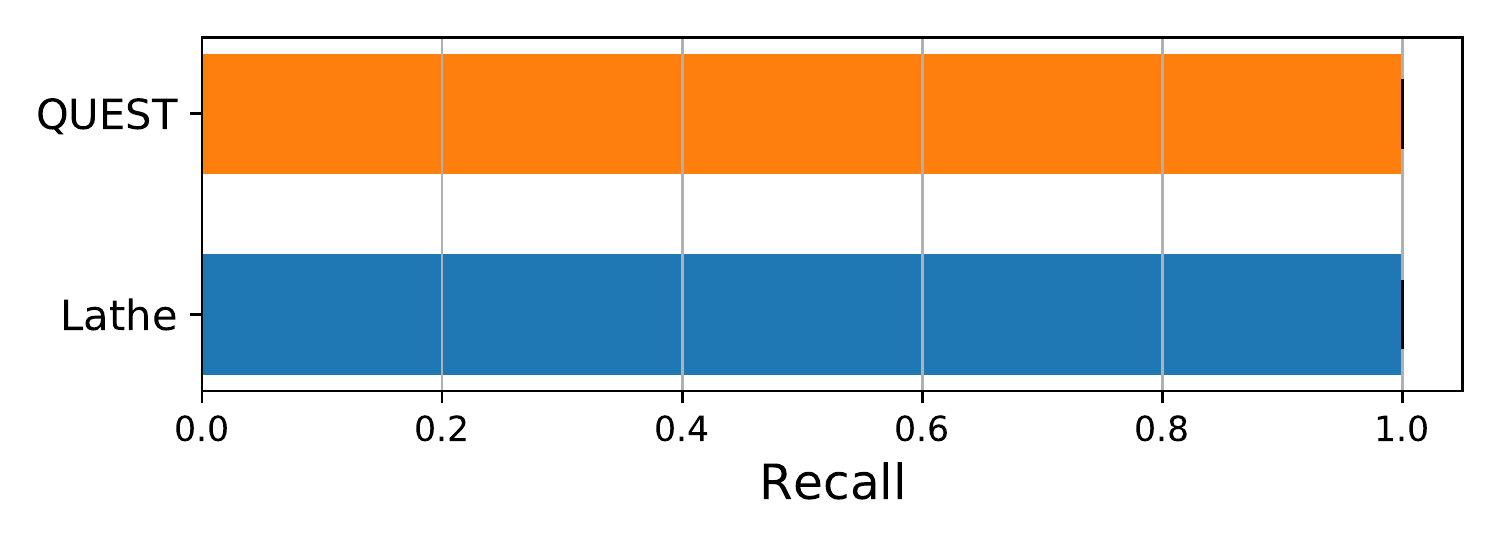}
         &
         \includegraphics[width=.5\textwidth]{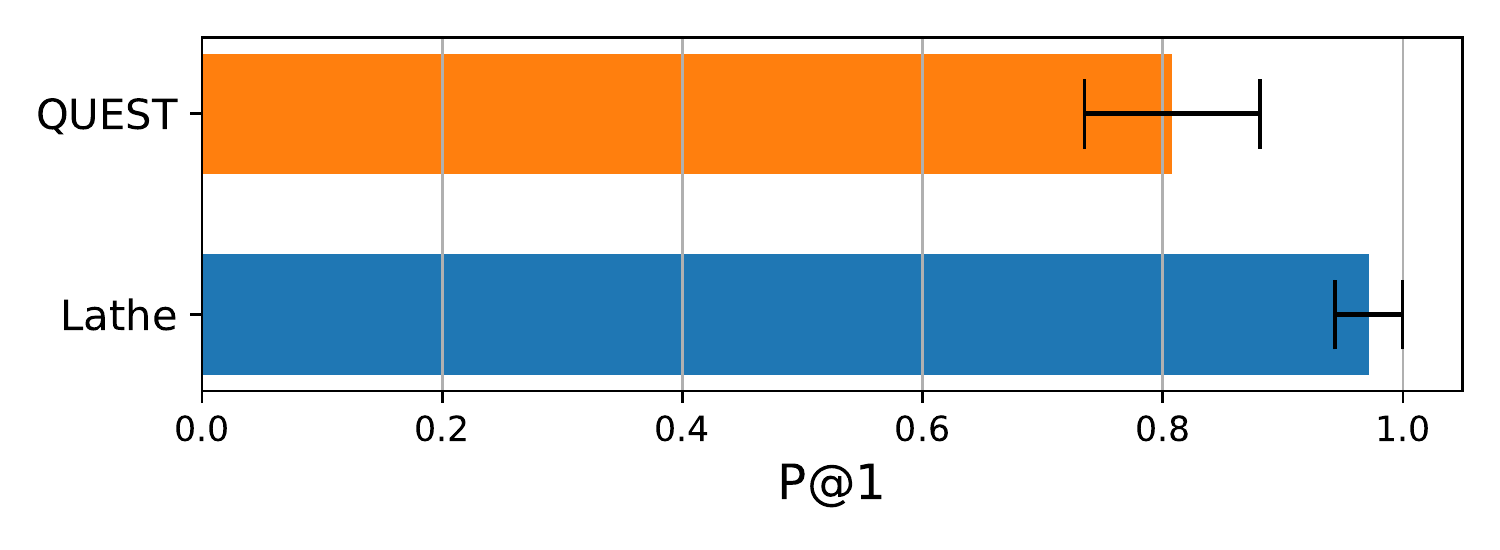}
    \end{tabular}
    \vspace*{-10pt}
    \caption{Comparison of {\metodo} with the QUEST system.}
    \label{fig:quest-comparison-35queries}
\end{figure}

Both systems achieved perfect recall; that is, all the correct solutions for the given keyword queries were retrieved. Concerning P@1, {\metodo} obtained better results than QUEST, with an average of 0.97 with a standard error of 0.03, which indicates that, in most cases, the correct solution
was the one corresponding to the CJN ranked as the first by {\metodo}.

Next, we compare the results obtained for Lathe with those published in the comprehensive evaluation published by Coffman \& Weaver~\cite{Coffman@CIKM10Framework} for the systems BANKS\cite{Aditya@VLDB02BANKS}, DISCOVER\cite{Hristidis@VLDB02DISCOVER}, DISCOVER-II\cite{Hristidis@VLDB03Efficient}, BANKS-II\cite{Kacholia@VLDB05Bidirectional}, DPBF\cite{Ding@ICDE07DPBF}, BLINKS\cite{He@SIGMOD07BLINKS} and STAR\cite{Kasneci@ICDE09STAR}. 
Because this comparison uses all 50 keyword queries from the MONDIAL dataset, we did not include QUEST in the comparison. Figure~\ref{plot:comparison-with-other-systems} shows the recall and P@1 values for the raking produced by each system when the golden standard provided by Coffman \& Weaver~\cite{Coffman@CIKM10Framework} is taken into account. 

\begin{figure*}[!htb]
     \centering
     \adjustbox{width=\textwidth}{
        \begin{tabular}{cc}
            \includegraphics[width=\linewidth]{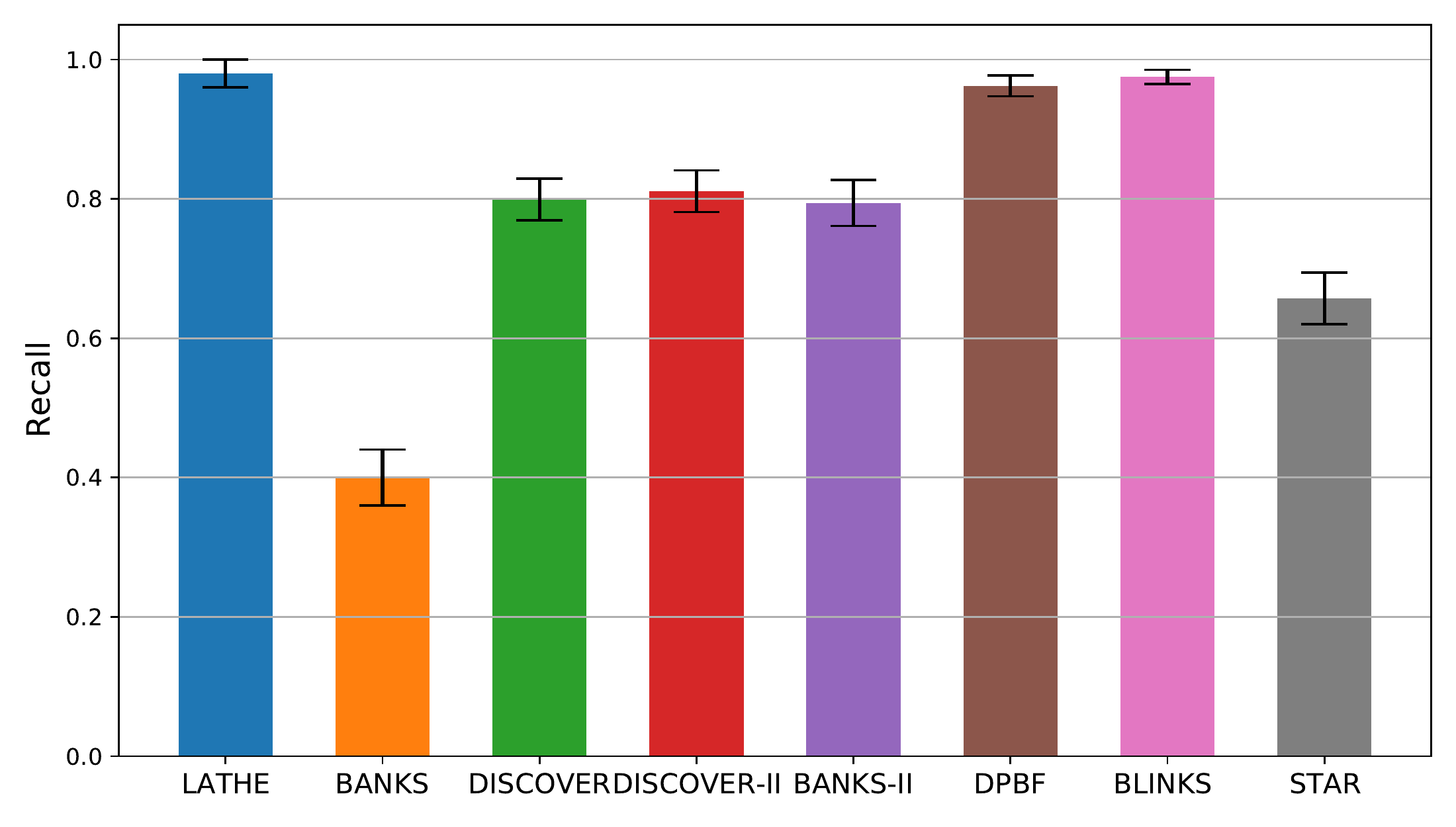}
              &
            \includegraphics[width=\linewidth]{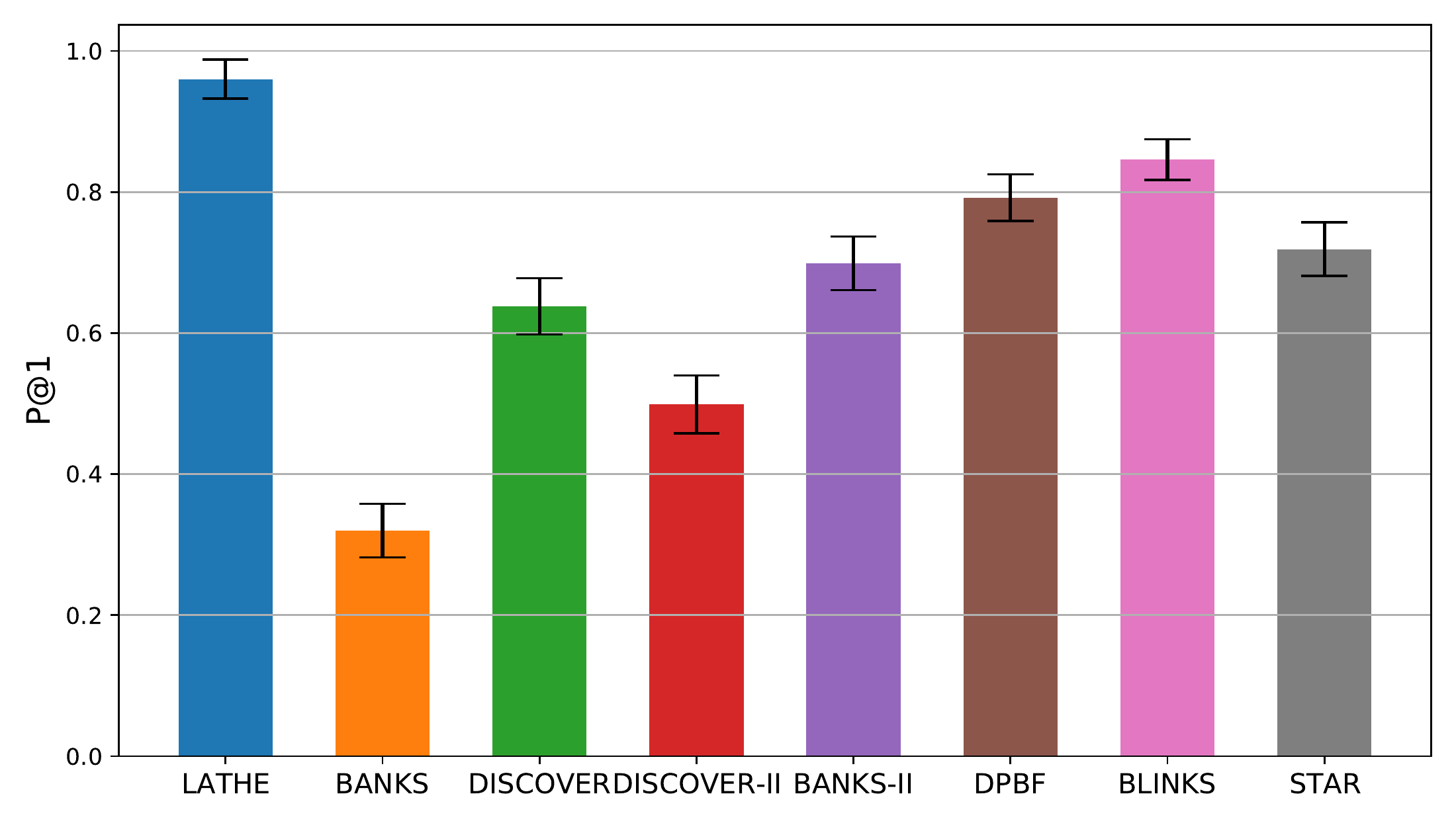}
        \end{tabular}
    }
    \caption{Comparison with other approaches using Recall and P@1 metrics.}
    \label{plot:comparison-with-other-systems}
\end{figure*}

Overall, Lathe achieved the best results in Recall and P@1 value. That the only systems that achieved similar recall, DPBF and BLINKS, are based on data graph, thus, require a materialization of the database. The difference between recall values of Lathe, DISCOVER, and DISCOVER-II is mainly due to not supporting schema references.
Regarding the P@1, Lathe obtained a value of 0.96 with a standard error of 0.03, which is significantly higher than the results for other systems. This difference in P@1 value, especially compared with DISCOVER and DISCOVER-II, is due to the novel ranking of QMs as well as an improved ranking of CJNs.

\subsection{Evaluation of Query Matches Ranking}

In this experiment, we evaluate the quality of QMs ranking according to the metrics MRR and $P@K$.
As shown by the results in Section~\ref{sec:gresults}, there can be many QMs depending on the query. 
As a result, we want to verify how effective the QMRank algorithm is at selecting the most likely correct QM from among those generated in this experiment. Figure~\ref{plot:ranking-query-matches}
shows the results obtained with $P@K$ up to the tenth ranking position and the MRR metric.

\begin{figure*}[!htb]
     \centering
     \adjustbox{width=\textwidth}{
        \begin{tabular}{m{\textwidth}m{\textwidth}}
             \includegraphics[width=\textwidth]{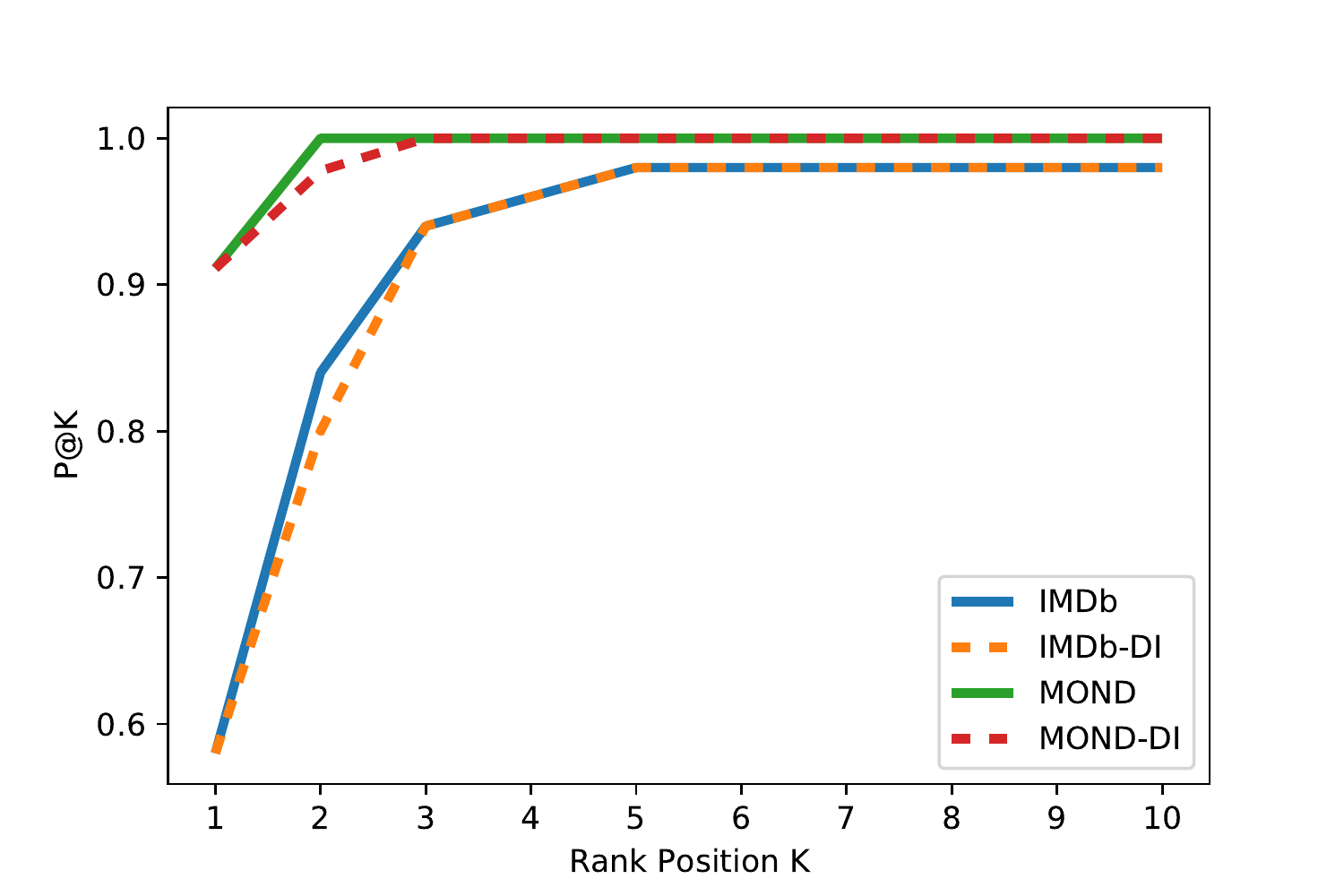}
             &
             \includegraphics[width=\textwidth]{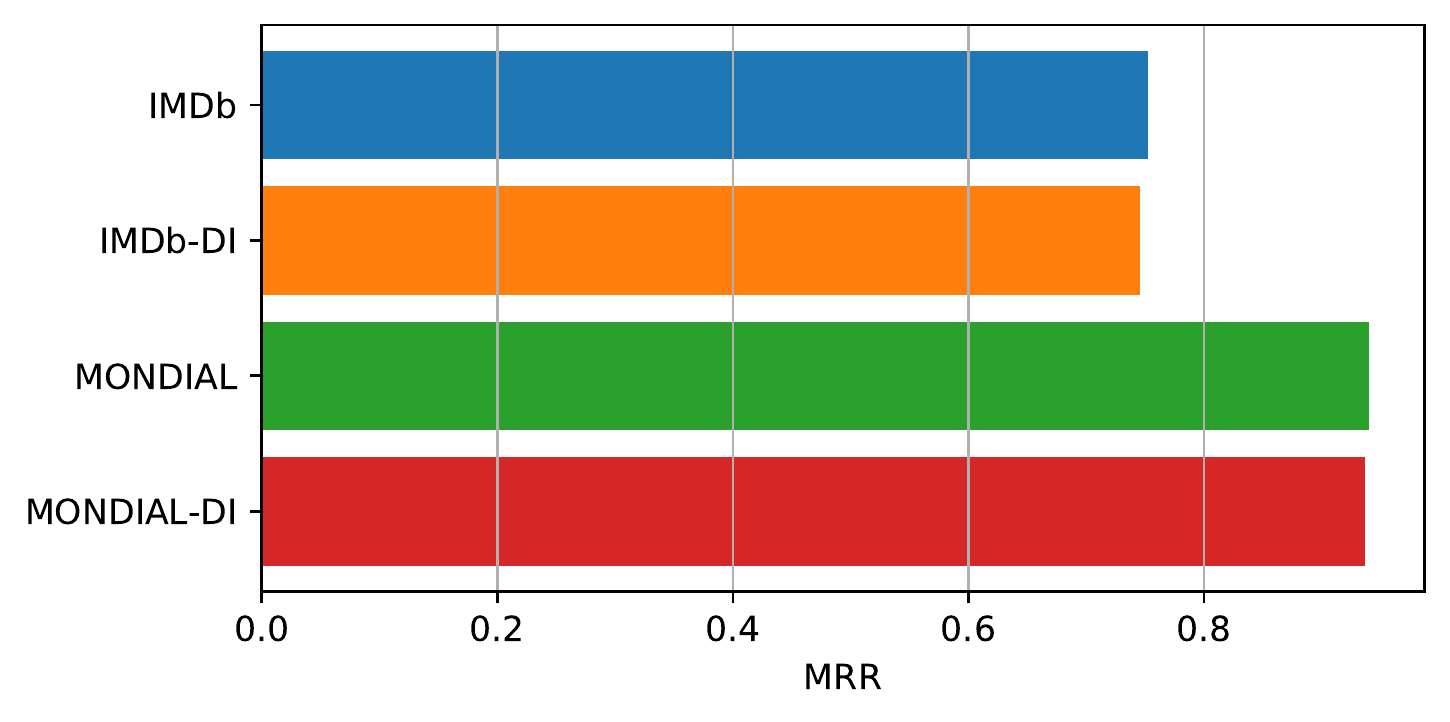}
        \end{tabular}
    }
    \vspace*{-10pt}
    \caption{Evaluation of Query Matches}
    \label{plot:ranking-query-matches}
\end{figure*}

For all query sets, in most cases, the correct QM is at least in the fifth ranking 
position since the peak result is from K=5 on. In MOND and MOND-DI, the relevant QM is at least in the third position for all queries. 
The IMDb and IMDb-DI curves never reached a P@K value of 1.0. 
This occurred because Lathe was unable to generate the correct QM for one of the 50 queries in these query sets. Thus, this QM is never ranked. 
This is an unusual query consisting of a full movie title with
26 terms that is highly ambiguous and was incorrectly matched
to the movie description by our system. In terms of MRR, {\metodo} obtained $0.75$ for both IMDb and IMDb-DI, as well as $0.96$ and $0.95$ for MOND and MOND-DI, respectively.
This indicates that the relevant QM is often found in the top positions of the ranking.
Notice that, the quality of the QM ranking indirectly impacts the ranking of CJNs.
Based these results, we can define the parameter $N_{QM}$, the maximum number of QMs 
to be selected before the CJN generation itself, and lead to use $N_{QM}=5$ in the default setup.


\subsection{Evaluation of the Candidate Joining Network Ranking} \label{sec:evaluation-cnrank}

In this experiment, we evaluate the quality of our approach for CJN generation and ranking. We used the metrics MRR and $P@K$ for $K$ up to the tenth rank position. We tested several different setups but to save space we report here only those with representative distinct results. Specifically, we report the results of four setups without the eager evaluation, that is, $5/1/0$, $5/2/0$, $5/8/0$ and $5/9/0$
and two setups with the eager evaluation, that is $5/1/9$ and $5/2/9$.

\begin{figure*}[!htb]
	\centering
	\includegraphics[width=\linewidth]{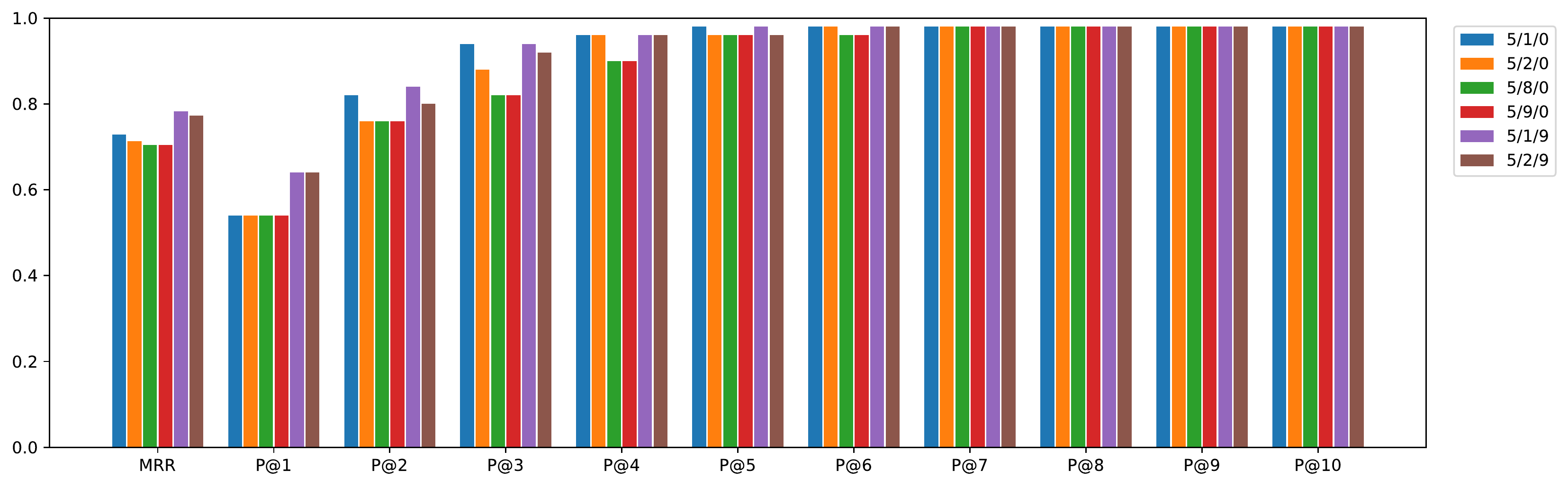}
	\includegraphics[width=\linewidth]{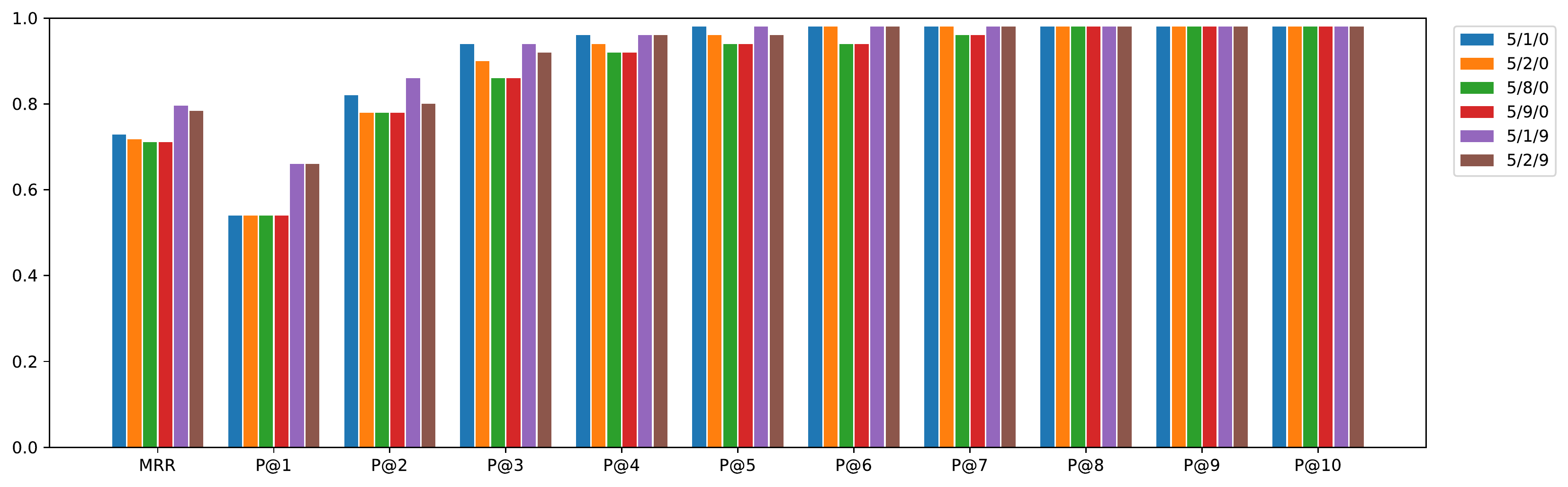}
	\caption{Ranking of Candidate Joining Networks - IMDb (top) and IMDb-DI (bottom)}
	\label{plot:ranking-cns-coffman-imdb}
\end{figure*}

Figure~\ref{plot:ranking-cns-coffman-imdb} shows the results for the IMDb and IMDb-DI query sets. As it can be seen, regardless of the configuration, our method was able to place the relevant CJNs in the top positions in the ranking, and the result is very similar for both IMDb and IMDb-DI query sets. This shows that in these datasets, our method was able to disambiguate the queries properly, even without the addition of schema references. 
It is worth noting that the values of P@1 in both datasets show that the configurations with the eager evaluation achieved better results because they place the relevant CJNs in the first ranking position more frequently. The $P@K$ metric also shows that
the quality of the ranking decreases as the number of CJNs per QM increases, especially for $K$ in the range $2\leq K \leq6$.

\begin{figure*}[!htb]
	\centering
	\includegraphics[width=\linewidth]{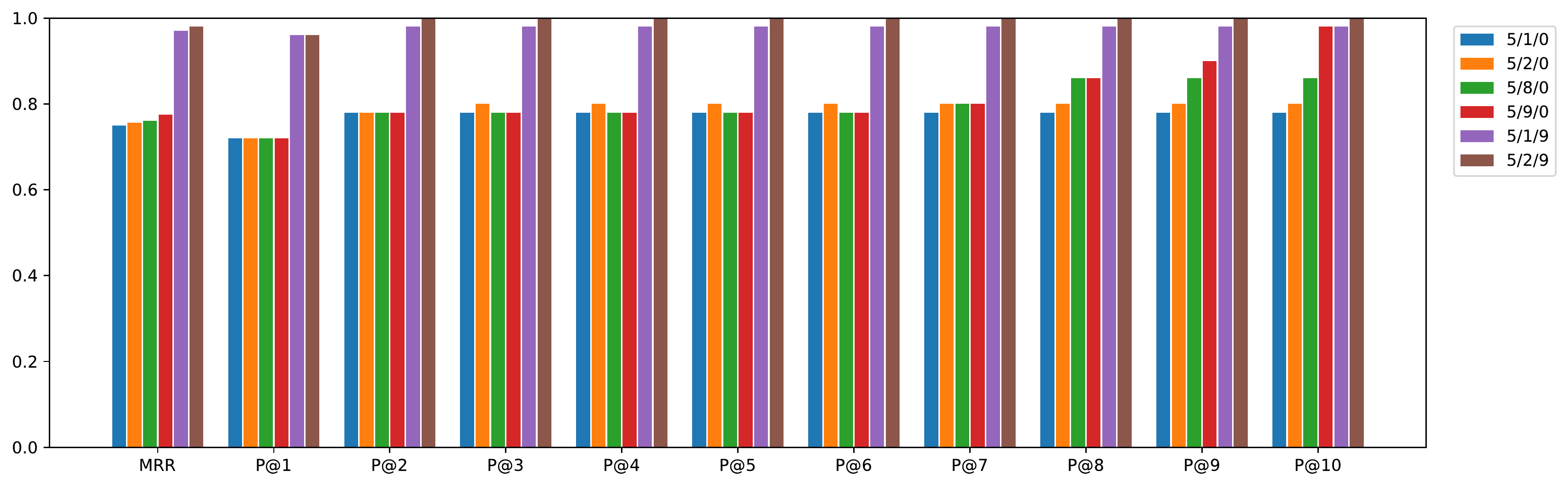}
	\includegraphics[width=\linewidth]{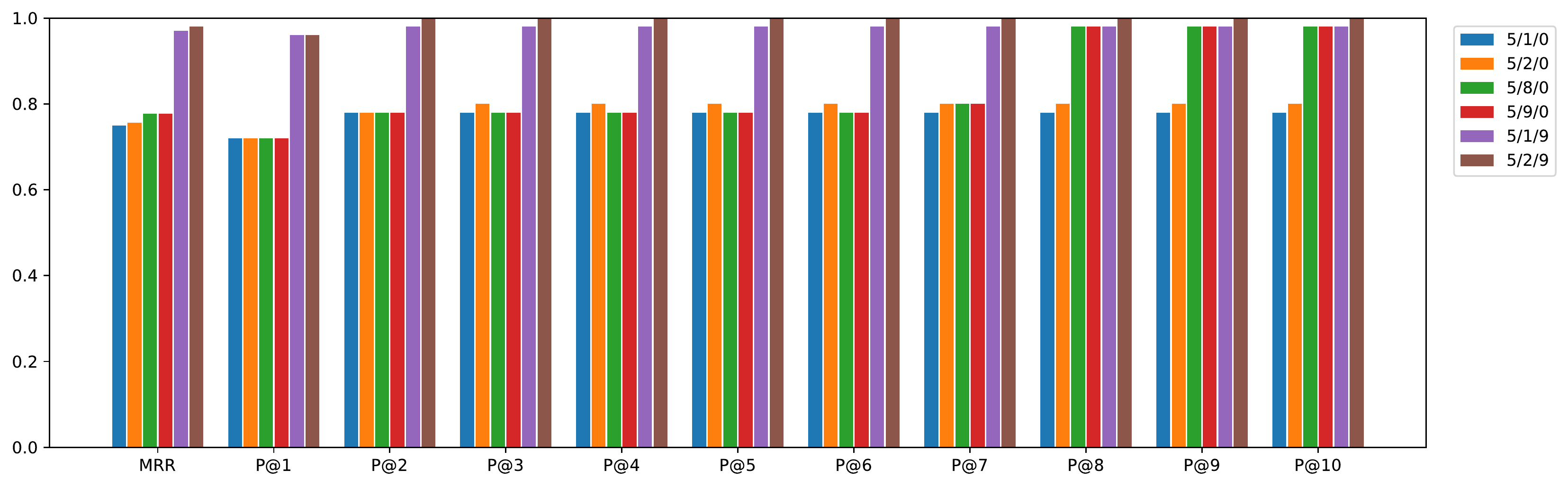}
	\caption{Ranking of Candidate Joining Networks - MONDIAL (top) and MONDIAL-DI (bottom)}
	\label{plot:ranking-cns-coffman-mondial}
\end{figure*}

Figure~\ref{plot:ranking-cns-coffman-mondial} shows the results for the MONDIAL and MONDIAL-DI query sets. In these datasets, the configurations with the eager evaluation achieved significantly better results. The configurations $5/1/0$ and $5/2/0$ could not generate the relevant CJN for around $20\%$ of the queries due to a low number of CJNs per QM, therefore, their results were capped at an MMR and $P@K$ value of $0.8$, approximately. The configurations $5/8/0$ and $5/9/0$ were able to generate the relevant CJN for most of the cases, although the large number of CJNs per QM negatively affected the ranking of CJNs. 
Finally, the configurations $5/1/9$ and $5/2/9$ produced the best results because the pruning enables us to generate the relevant CJN with a low number of CJNs per QM while also placing the relevant CJN in higher rank positions. Notice that the disambiguation of queries in the MONDIAL-DI query set allowed configurations $5/8/0$ and $5/9/0$ to have better results, especially for the $P@K$ metric for $K$ above 8. The eager evaluation configurations were able to disambiguate the queries without relying on the addition of schema references, therefore, their results were consistent across the MONDIAL and MONDIAL-DI query sets.

Regardless of the datasets and configurations, our method achieved an MRR value above 0.6, which indicates that on average, the relevant CJN is found between the first and the second rank positions. In the IMDb dataset, the decrease of $P@K$ values according to the number of CJNs taken per QM is also reflected on the MRR metric. However, in the MONDIAL dataset, the improvement of the $P@K$ values due to the disambiguation of queries is not reflected on the MRR value, as this improvement only happens in low ranking positions ($K \leq 8$).  

The eager CJN evaluation inherently affects the performance of the CJN generation process. Therefore it is important to look at the trade-off between the effectiveness and the efficiency in each configuration. We examine this trade-off in the next section.

\subsection{Performance Evaluation}

In this experiment, we aim at evaluating the time spent for obtaining the CJN given a keyword query, and analyze the trade-offs between efficiency and efficacy of the different configurations use in {\metodo}.

\begin{figure}[!htb]
    \centering
    \includegraphics[width=0.65\textwidth]{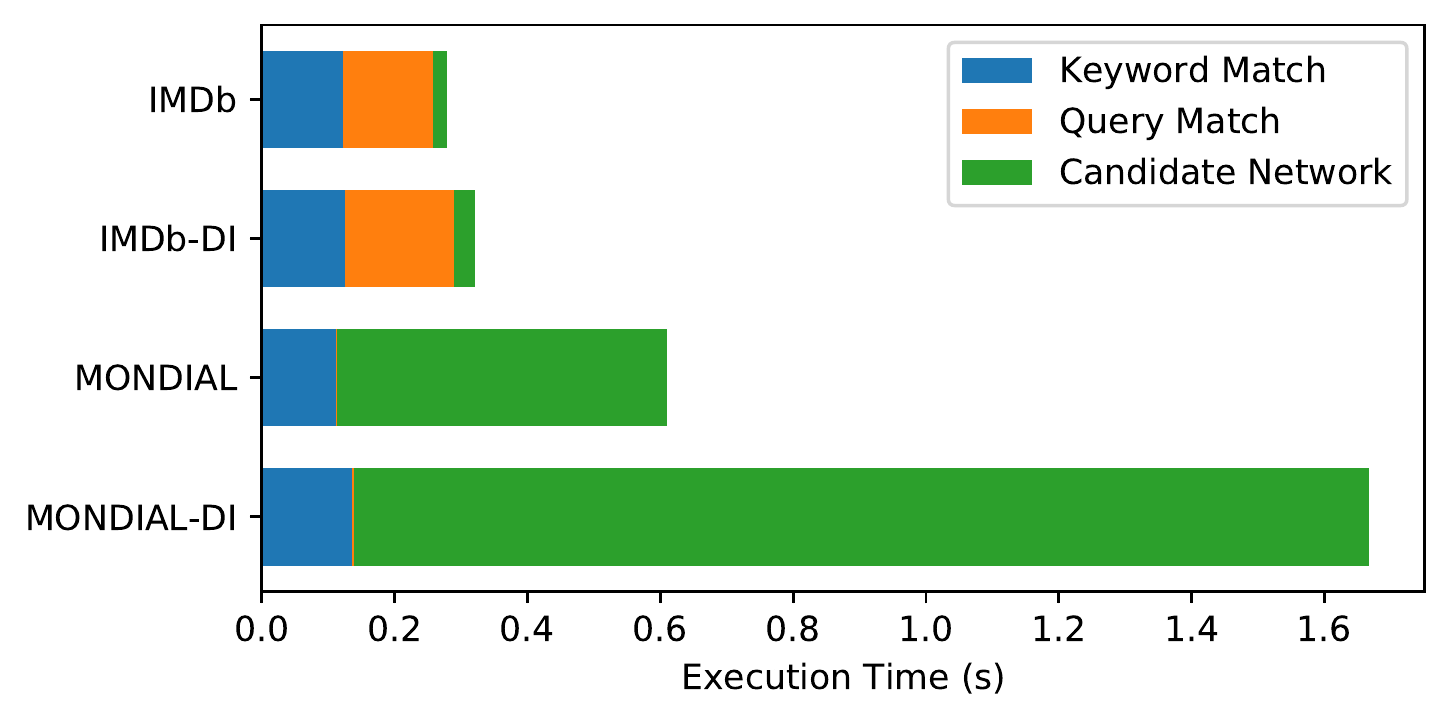}
    \caption{Average Execution Times for each phase of {\metodo}}
    \label{plot:performance-execution}
\end{figure}

Figure~\ref{plot:performance-execution} summarizes the average execution time for each phase of the process: Keyword Matching, Query Matching and the Candidate Joining Network Generation. In this first experiment, we used the configuration $5/1/0$. {\metodo} obtained better execution times for the IMDb dataset in all configurations. Also, the disambiguate variants of query sets 
yield slower execution times in comparison with the original counterparts. 

It is worth noting that the execution times for each query set are proportional to the number of KMs, QMs, and CJNs in the query sets shown in Table~\ref{tab:general-results}. 
Despite the lower number of KM generated in the case of MONDIAL and MONDIAL-DI, the execution times are similar for all query sets during the Keyword Matching phase.
One explanation for that is the high number of schema elements in the MONDIAL dataset, since there are 28 relations and 48 attributes to match when looking for SKMs (see Table~\ref{tab:datasets}). 
Due to the combination nature of QM generation, the execution times for the Query Matching phase are directly related to the number of QMs.
While the execution times for the IMDb and IMDb-DI query sets that produced a
high number of QMS are $0.13$ and $0.16$ seconds, respectively, the results for MONDIAL and MONDIAL-DI are around $1$ and $3$ milliseconds. 
Concerning the CJN phase, the execution times for MONDIAL are significantly higher in comparison with the execution times for IMDb, despite the lower number of CJNs for the former dataset. 
Because the CJN generation algorithm is based on a Breadth-First Search, the greater the number of vertices and edges in the schema graph of the MONDIAL dataset, the greater the number of iterations and, consequently, the slower the execution times. This behavior persists throughout different configurations, an issue we further analyze below.

\subsection{Quality versus Performance}

Figure~\ref{plot:performance-cn-generation} presents an evaluation of the CJN generation performance, comparing the same configurations used in the experiment of Section~\ref{sec:evaluation-cnrank}. We present the results for the MONDIAL and IMDb datasets in different scales because they differ by order of magnitude. Overall, execution times increase as the number of CJNs taken per QM increases. This pattern is more pronounced in the MONDIAL dataset. 
Because the system has to probe the CJNs running queries into the database, the eager CJN evaluation incurs an unavoidable increase of the CJN generation time.

\begin{figure*}[!htb]
    \centering
    \adjustbox{width=\textwidth}{
        \begin{tabular}{rl}
            \includegraphics[height=140px]{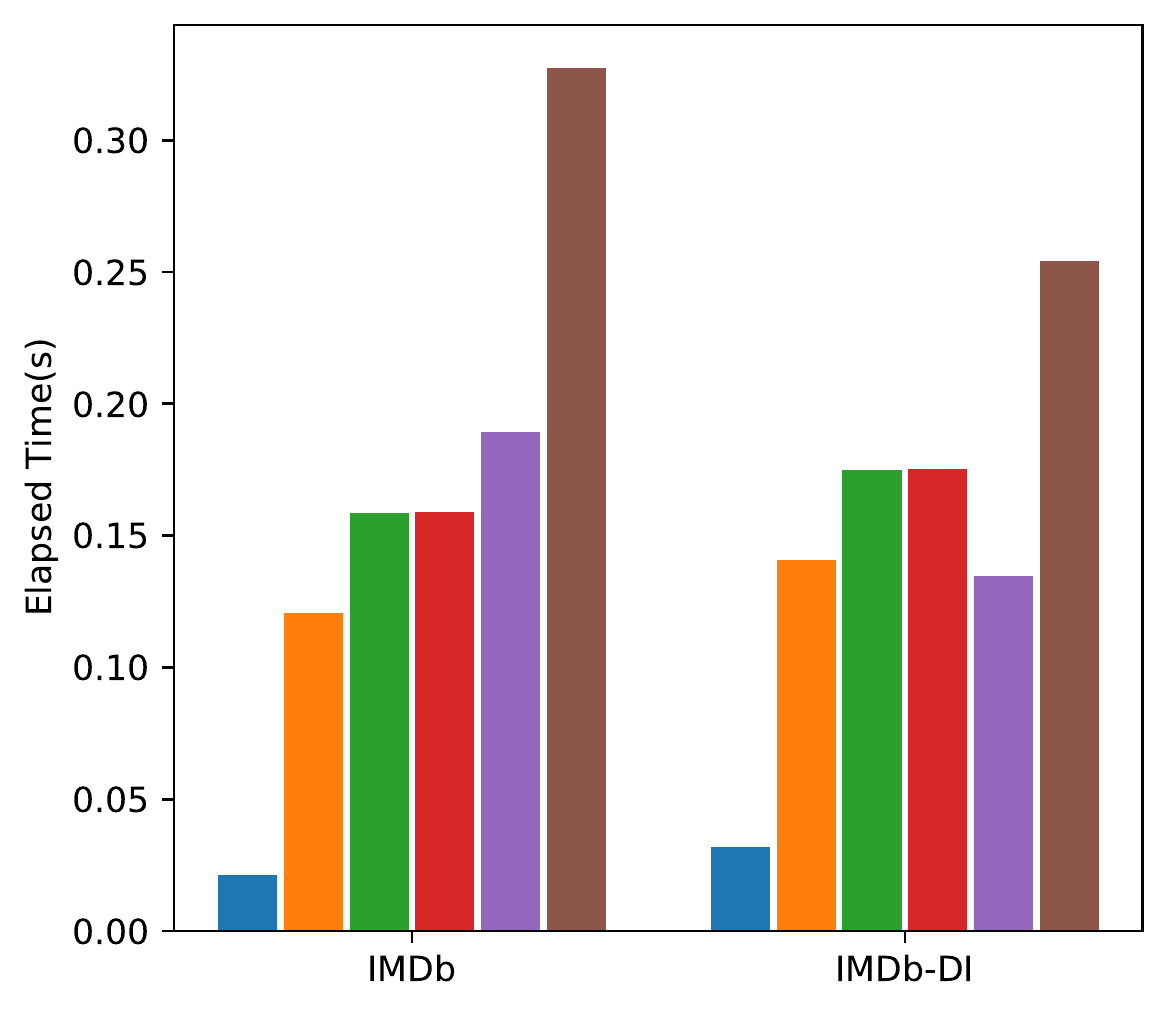}
            & 
            \includegraphics[height=140px]{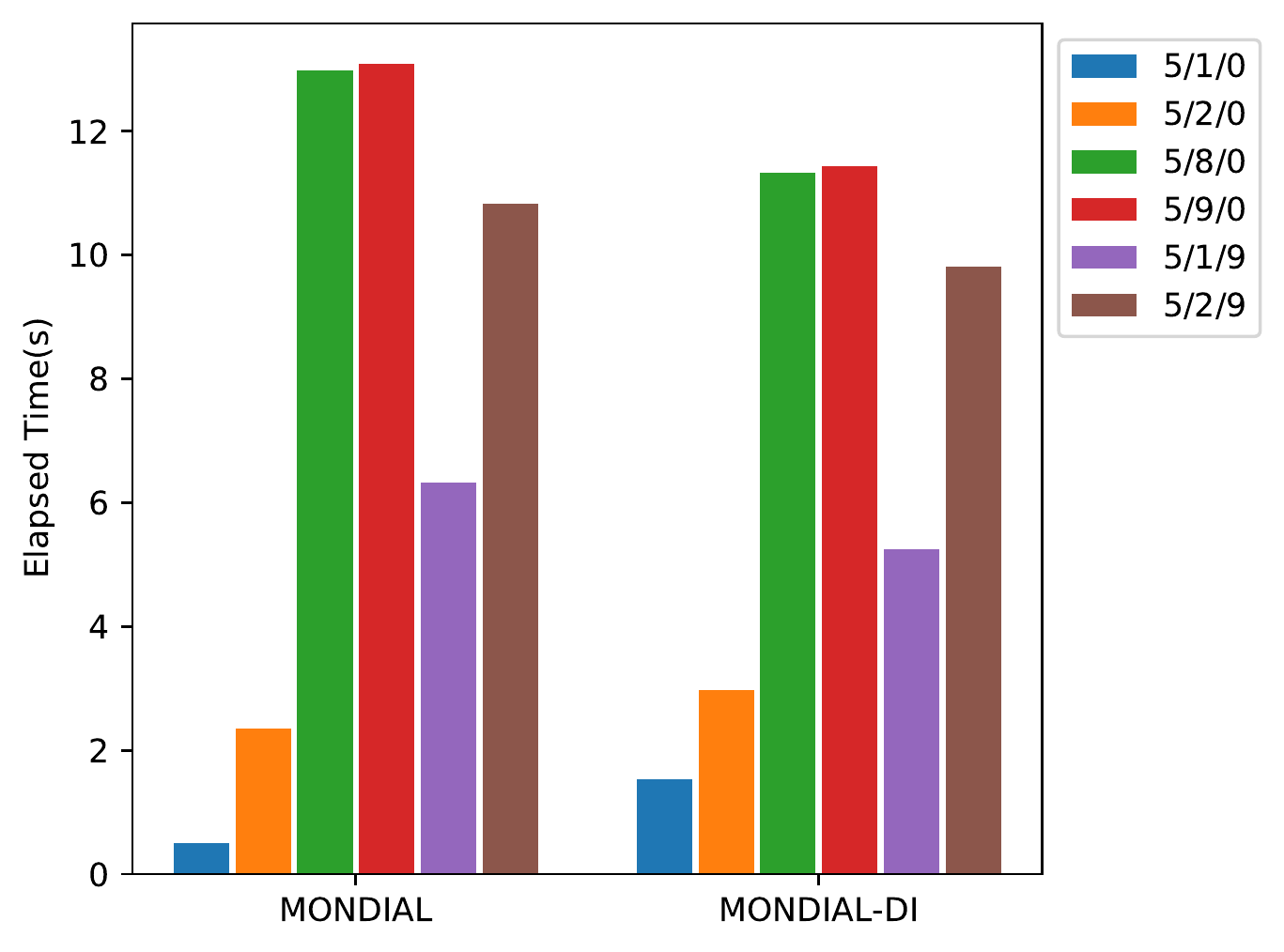}
        \end{tabular}
    }
    \caption{Performance Evaluation of the CJN Generating phase}
    \label{plot:performance-cn-generation}
\end{figure*}

As the configurations have an impact on both the quality of the CJN ranking and the performance, it is important to examine the trade-off between  effectiveness and efficiency. Configurations $5/1/0$ and $5/2/0$ achieved the best execution times due to the low number of CJNs per QM and not relying on database accesses. However, these configurations did not achieve the highest values of MRR and P@K. Therefore, they are recommended only if efficiency must be prioritized.
Configurations $5/8/0$ and $5/9/0$ achieved the worst execution times for the MONDIAL dataset and intermediate execution times for the IMDb dataset due to the high number of CJNs per QM. Although these configurations achieved a slight increase in the $P@K$ metric for the MONDIAL dataset, for $K \leq 8$, the values for both MRR and $P@K$ values for the IMDb dataset were the lowest among the configurations. 

Regarding the configurations $5/1/9$ and $5/2/9$, they achieved the worst execution times for the IMDb dataset and intermediate execution times for the MONDIAL dataset due to the eager CJN evaluation. However, the significantly better results of MRR and $P@K$ values for these configurations make them recommended options, if effectiveness must be prioritized. Notice that $5/1/9$ achieved better results than $5/8/0$ and $5/9/0$ for all query sets except the IMDb query set, in which the difference of execution times is less than 50ms. As for configuration $5/2/9$, the execution times were the worst for the IMDb dataset, but the results for the MONDIAL dataset were better than $5/8/0$ and $5/9/0$. While the MRR and $P@K$ values for configurations $5/1/9$ and $5/2/9$ were close, the former has better performance. Thus the configuration $5/1/9$ is recommended overall, especially if effectiveness must be prioritized.

It is interesting noting that although the configurations with eager CJN evaluation spend time to probe CJNs, sending queries to the DBMS. However, as they generate a smaller set of CJNs, the overall performance is not hindered in comparison with the configurations without it.

%% file: tables/tab-datasets.tex
\begin{tabular}{lrrrrr} 
	\toprule
	Dataset & Size(MB) & Relations &  Attributes & RIC & Tuples     \\ 
	\midrule
	MONDIAL & 13        & 28  & 48       & 38  & 17,115     \\
	IMDb    & 528      &    6  & 33     & 5   & 1,673,076 
\\
	\bottomrule
\end{tabular}

%% file: tables/tab_general_results.tex
\begin{tabular}{|l|r|r|r|r|r|r|}
\hline
~ & \multicolumn{2}{c|}{\textbf{Num. KMs}} & \multicolumn{2}{c|}{\textbf{Num. QMs}}  & \multicolumn{2}{c|}{\textbf{Num. CNs}}   \\ \hline
\multicolumn{1}{|l|}{Query sets} & \multicolumn{1}{c|}{Max} & \multicolumn{1}{c|}{Avg} & \multicolumn{1}{c|}{Max} & \multicolumn{1}{c|}{Avg} & \multicolumn{1}{c|}{Max} & \multicolumn{1}{c|}{Avg}  \\ \hline
\multicolumn{1}{|l|}{IMDb}    & 47   & 15.38 & 515 & 72.46 & 554 & 89.04 \\ \hline
\multicolumn{1}{|l|}{IMDb-DI} & 64   & 16.20 & 515 & 69.76 & 492 & 85.28 \\ \hline
\multicolumn{1}{|l|}{MOND}    & 8    & 3.16  & 9   &  2.11 & 33  & 6.96  \\ \hline
\multicolumn{1}{|l|}{MOND-DI} & 8    & 3.47  & 9   &  2.47 & 34  & 7.87  \\ \hline
\end{tabular}

%% file: sections/8-conclusions.tex

In this paper we have proposed {\metodo}, a new relational keyword search (R-KwS) system for generating a suitable SQL query from a given keyword query. 
{\metodo} is the first  to address the problem of generating and ranking Candidate Joining Networks (CJNs) based on queries with keywords that can refer
to either instance values or database schema elements, such as relations and attributes. In addition, {\metodo} improves the quality of 
the CJN generated by introducing two major innovations: a ranking for selecting better Query Matches (QMs) in advance, yielding the generation
of fewer but better CJNs, and an eager evaluation strategy for pruning void useless CJNs. 
We present a comprehensive set of experiments performed with query 
sets and datasets previously used in experiments with previous state-of-the-art R-KwS systems and methods.
Our experiments indicate that {\metodo} can handle a 
wider variety of keyword queries while remaining highly effective, even for large databases with intricate schemas.

\newcommand{\footmongo}{\footnote{\url{https://www.mongodb.com/} }}

Our experience in the development of {\metodo} raised several ideas for future work.
First, one important issue in our method is being able to correctly match keywords from the input query to the corresponding database elements. To improve this issue, we plan to investigate new alternative similarity functions. 
We are particularly interested in using word-embedding-based functions, such as the well-known Word Mover’s Distance (WMD) \cite{Kusner@PMLR15WMD}.
Second, data exploration techniques have recently gained popularity because they allow for the extraction of knowledge from data even when the user is unsure of what to look for \cite{Idreos@SIGMOD15DataExploration}. 
Keyword-based queries, we believe, can be used as an interesting tool for data exploration because they allow one to retrieve interesting portions of a database without knowing the details of the schema and its semantics.  
Third, although we have focused on relational databases in this paper, the ideas we discussed here can be extended to other types of databases as well. Currently, we are extending these ideas to address the so-called \emph{document stores}, such as the very popular MongoDB{\footmongo} engine. 
Our preliminary findings suggest that because queries of this type are frequently more complex than queries of relational databases, allowing the simplicity of keyword queries may have even more advantages in this context. 
Finally, we anticipate that keyword queries will be useful as a tool for allowing the seamless integration of data from heterogeneous sources, as is the case in the so-called polystores systems and data lakes, which are becoming increasingly popular in recent years. There exist already research proposals in this direction \cite{Chanial@VLDB18Connectionlens}, we believe that the schema graph approach we adopt in our work can be helpful to achieve this goal.

%% file: sections/9-apendices.tex

\appendix

\section{VKMGen Algorithm} \label{apx:vkmgen}

\input{appendices/apx-a-vkmgen}

\section{SKMGen Algorithm} \label{apx:skmgen}

\input{appendices/apx-b-skmgen}

\section{QMGen Algorithm} \label{apx:qmgen}

\input{appendices/apx-c-qmgen}

\section{QMRank Algorithm} \label{apx:qmrank}

\input{appendices/apx-d-qmrank}

\section{Sound Theorem} \label{apx:sound-thm}

\input{appendices/apx-e-soundtheorem}

\section{CJNGen Algorithm} \label{apx:cnkmgen}

\input{appendices/apx-f-cnkmgen}

%% file: appendices/apx-a-vkmgen.tex

As shown in Algorithm~\ref{alg:vkmgen},  {\metodo} retrieves tuples from the database in which the keywords occur and uses them to generate value-keyword matches.
Initially, the VKMGen Algorithm takes the occurrences of each keyword from the \textit{Value Index} and form partial value-keyword matches, which are not guaranteed to be disjoint sets yet (Lines \ref{line:vkmgen-begin-for}-\ref{line:vkmgen-conversion}). The pool of VKMs is represented by the \emph{Hash Table} $P$, whose keys are KMs and values are sets of tuple IDs. 

\input{algorithms/alg-VKMGen}

Next, {\metodo} ensures that VKMs are disjoint sets through the Algorithm~\ref{alg:VKMInter}, VKMInter, which is based on the ECLAT algorithm~\cite{Zaki@IEEE00Eclat} for finding frequent itemsets. VKMInter looks for non-empty intersections of the partial value-keyword matches recursively until all of them are disjoint sets, and thus, proper VKMs. These intersections are calculated as follows:
\[
K\!M_1 \cap K\!M_2 =
	\begin{dcases}
	\emptyset &,\text{if }R_a\neq R_b\\
	R^{V}_{ab}[A_{ab,1}^{K_{ab,1}},\ldots,A_{ab,m}^{K_{ab,m}}]&, \text{if }R_a=R_b\\
	\end{dcases}	
\]
where $K\!M_x = R^{V}_x[A_{x,1}^{K_{x,1}},\ldots,A_{x,m}^{K_{x,m}}]$ for $x \in \{a,b\}$, and $K_{ab,i} = K_{a,i}\cup K_{b,i}$.

\input{algorithms/alg-DelayedVKMInter}

VKMInter uses three hash tables: $P$, $P_{next}$ and $R$. The pool $P$ contains the partial VKMs of the current iteration. The pool $P_{next}$ contains the partial VKMs for the next iteration. The pool $R$ stores the tuple IDs to be removed from the VKMs of $P$ at the end of the current iteration, turning the partial VKMs into proper value-keyword matches. 

VKMInter first defines the hash tables $P_{next}$ and $R$, then initializes $R$ with empty sets (Lines \ref{VKMInter:line-begin-definition}-\ref{VKMInter:line-end-definition}). Next, the algorithm iterates over all pairs $\{K\!M_{a},K\!M_{b}\}$ of VKMs in $P$ and tries to create a new keyword match $K\!M_{ab}$, which is the intersection of $K\!M_{a}$ e $K\!M_{b}$ (Lines \ref{VKMInter:line-begin-intersection}-\ref{VKMInter:line-end-intersection}).  If $K\!M_{ab}$ is valid, that is, if $K\!M_{a}$ e $K\!M_{b}$ are VKMs over the same database relation, and the tuples $T_{ab}$ within $K\!M_{ab}$ are not empty, then we add $K\!M_{ab}$ to the next iteration pool $P_{next}$ and add the tuples $T_{ab}$ to $R$ for removal after the iteration (Lines \ref{VKMInter:line-if-intersection}-\ref{VKMInter:line-add-for-removal}). After all the possible intersections are processed, VKMInter iterates over $R$ and removes the tuples for each VKM of the pool $P$, making them proper disjoint keyword matches (Lines \ref{VKMInter:line-begin-removal}-\ref{VKMInter:line-end-removal}). Lastly, VKMInter recursively process the pool $P_{next}$ for the next iteration, then it updates and returns the current pool $P$ (Lines \ref{VKMInter:line-update-pool}-\ref{VKMInter:line-return}).

After the execution of VKMInter, in Line \ref{VKMGen-line-recursive-call} of VKMGen, we obtained the value-keyword matches and their tuples. As the sets of tuples are only required for the generation of VKMs, VKMGen generates and outputs the set of value-keyword matches, ignoring the tuples from $P$ (Lines \ref{VKMGen-line-begin-ignore-tuple}-\ref{VKMGen-line-end-ignore-tuple}). From now on, {\metodo} does not need to manipulate the database tuples or their IDs.

%% file: algorithms/alg-VKMGen.tex
\begin{algorithm}[!htb]
	\caption{VKMGen($Q$)}
	\label{alg:vkmgen}
	\KwIn{A keyword query $Q{=}\{k_{1},k_{2},\ldots,k_{m}\}$}
	\KwOut{The set of value-keyword matches $V\!K$}
	
	\Let{$I_V$ the Value Index}
    \Let{$P\!$ be a Hash Table.}
	\For{keyword $k_i \in Q$}{ \label{line:vkmgen-begin-for}
		\If{$k_i \in I_V$}{
			\For{relation $R_j \in I_V[k_i]$}{
				\For{ attribute $A_k \in I_V[k_i][R_j]$}{
				    \Let{$K\!M$ be the partial keyword match $R_j^V[A_k^{\{ k_i \}}]$}
					$P[K\!M] \leftarrow I_V[k_i][R_j][A_k]$ \label{line:vkmgen-conversion}
				}
			}	
		}
	}
	$P \leftarrow \textbf{VKMIter}(P$) \; \label{VKMGen-line-recursive-call}
	\For{value-keyword match $K\!M_{u} \in P$}{ \label{VKMGen-line-begin-ignore-tuple}
	    $V\!K \leftarrow V\!K \cup \{K\!M_{u}\}$ \;
	} \label{VKMGen-line-end-ignore-tuple}
	\Return $V\!K$	
\end{algorithm}

%% file: algorithms/alg-DelayedVKMInter.tex
\begin{algorithm}[!htb]

	\caption{VKMInter($P$)}
	\label{alg:VKMInter}

	\KwIn{A Hash Table $P$ whose keys are partial value-keyword matches and values are tuples.}
	\KwOut{A Hash Table $P$ whose keys are proper value-keyword matches and values are tuples.}

    \Let{$P_{next}$ be a Hash Table.} \label{VKMInter:line-begin-definition}
    \Let{$R$ be a Hash Table.}
    \For{value-keyword match $K\!M_{u} \in P $}{
	    $R[K\!M_{u}] \leftarrow \emptyset $ \;
	} \label{VKMInter:line-end-definition}
	\For{pair of keyword matches $\{K\!M_a,K\!M_b\} \in \binom{P}{2}$}{ \label{VKMInter:line-begin-intersection}
	    $K\!M_{ab} \leftarrow K\!M_{a} \cap K\!M_{b}$ \;
		$T_{ab} \leftarrow P[K\!M_a] \cap P[K\!M_b]$ \;
		\If{$T_{ab} \neq \emptyset \ \KwAnd \ K\!M_{ab}$ is valid}{ \label{VKMInter:line-if-intersection}
			$P_{next}[K\!M_{ab}] \leftarrow T_{ab}$\;
			$R[K\!M_{a}] \leftarrow R[K\!M_{a}] \cup T_{ab}$\;
			$R[K\!M_{b}] \leftarrow R[K\!M_{b}] \cup T_{ab}$\; \label{VKMInter:line-add-for-removal}
		}
	} \label{VKMInter:line-end-intersection}
	\For{value-keyword match $K\!M_{u} \in R$}{ \label{VKMInter:line-begin-removal}
	    $P[K\!M_{u}] \leftarrow P[K\!M_{u}] - R[K\!M_{u}]$ \;
	    \If{$P[K\!M_{u}] = \emptyset $}{
	        \textbf{remove } $K\!M_{u}$ from $P$\;
	        $P.remove(K\!M_{u})$ \;
	    } 
	} \label{VKMInter:line-end-removal}
	$P_{next} \leftarrow \textbf{VKMInter}(P_{next})$\; \label{VKMInter:line-next-iteration}
    \textbf{update } $P$ \textbf{with} $P_{next}$\;    \label{VKMInter:line-update-pool}
	\Return{$P$} \label{VKMInter:line-return}
\end{algorithm}

%% file: appendices/apx-b-skmgen.tex

The generation of schema-keyword matches uses a structure we call the \textit{Schema Index}, which is created in a preprocessing phase, alongside with the Value Index. This index stores information about the database schema and statistics about attributes, which are used for the ranking of QMs, which will be explained in Chapter~\ref{chap:query-matching}. The stored information follows the structure below: \[I_S=\{relation:\{attribute:\{(norm,maxfrequency)\}\}\}\]

The generation of SKMs is carried out by Algorithm~\ref{alg:skmgen}, SKMGen. First, the algorithm iterates over the relations and attributes from the Schema Index. Then, SKMGen calculates the similarity between each keyword and schema element. It only considers the pairs with similarity above a defined threshold $\varepsilon$ (Line \ref{line:skmgen-threshold}), which are used to generate SKMs (Line \ref{line:skmgen-end-iteration}). 

\input{algorithms/alg-SKMGen}


%% file: algorithms/alg-SKMGen.tex
\begin{algorithm}[!htb]
	\caption{SKMGen($Q$)}
	\label{alg:skmgen}
	\KwIn{A keyword query $Q{=}\{k_{1},k_{2},\ldots,k_{m}\}$, the Schema Index $I_S$}
	\KwOut{The set of schema-keyword matches $S\!K$}
	$S\!K \leftarrow \{\}$\;
	\For{keyword $k_i \in Q$}{  \label{line:skmgen-start-iteration}
		\For{relation $R_j \in I_S$}{  \label{line:skmgen-end-iteration}			
			
			\If{$sim(k_i,R_j) \geq \varepsilon$}{ \label{line:skmgen-threshold1}
			    \Let{$K\!M$ be the schema-keyword match $R_j^S[self^{\{ k_i \}}]$}
				$S\!K\! \leftarrow S\!K\! \cup \{K\!M\}$)	
			}	
			
			\For{ attribute $A_l \in I_S[R_j]$}{
				\If{$sim(k_i,A_l) \geq \varepsilon$}{ \label{line:skmgen-threshold}
				    \Let{$K\!M$ be the schema-keyword match $R_j^S[A_l^{\{ k_i \}}]$}
					$S\!K\! \leftarrow S\!K\! \cup \{K\!M\}$
				}	
			}
			
		}	
	}
	\Return{$S\!K$}
\end{algorithm}

%% file: appendices/apx-c-qmgen.tex

The generation of query matches is carried out by Algorithm~\ref{alg:qmgen}, QMGen, which preserves the ideas
proposed in MatCNGen~\cite{Oliveira@ICDE18MatCNGen}, adapt them to keyword matches instead of tuple-sets.
Let $V\!K$ and $S\!K$ be respectively sets of value-keyword matches, and schema-keyword matches previously generated.
The algorithm looks for combinations of keyword matches in $P{=}V\!K{\cup}S\!K$ that form minimal covers for the query $Q$. 
At a first glance, this statement may suggest that we need to generate the whole power set of $P$ to obtain the complete set of QMs. However, it can be shown that any minimal cover of a set of $n$ elements has at most $n$ subsets~\cite{Hearne@73minimalCover}. 
Therefore, no match for a query $Q$ can be formed by more than $|Q|$ keyword matches. 
Also, as the QM ranking presented in Section~\ref{sec:query-match-ranking} penalizes QMs with a large number of KMs, we can define a maximum QM size $t{\leq}|Q|$ to prune QMs which are less likely to be relevant.
For this reason, QMGen iterates over all the subsets of $P$  whose size is less than or equal to a maximum QM size $t$, which is at most the size of the query $Q$ (Lines \ref{line:qmgen-iteration-max-size}-\ref{line:qmgen-iteration-combination}). Next, QMGen checks whether the combination $M$ of keyword matches form a minimal cover for the query. The evaluation of minimal cover is carried out by Algorithm~\ref{alg:minimalCover}.

\input{algorithms/alg-QMGen}

The algorithm MinimalCover iterates through the KMs from the combination $M$, generating a set $C_{M}$ which comprise all keywords covered by $M$ (Lines \ref{line:minimalcover-begin-cover-qm}-\ref{line:minimalcover-end-cover-qm}). Next, the algorithm checks whether $M$ is total, that is, whether $C_{M}{=}Q$. Notice that since KMs can only associate an attribute or relation in the database schema to keyword from the query $Q$, that is $C_{M}{\subseteq}Q$, then we can imply that $C_{M}{=}Q$ if, and only if, $|C_{M}|{=}|Q|$ (Line \ref{line:minimalcover-check-total}). Next, MinimalCover checks whether $M$ is minimal, that is, if we remove any keyword match from $M$ it will no longer be total. For this reason, MinimalCover iterates again through the KMs and, for each one, it generates a set $C_{KM}$ which comprise all keywords covered by $KM$. Then, the algorithm check whether the set difference of $C_{M}{\setminus}C_{KM}$ is still equal to $Q$, which can be achieved by comparing $|C_{M}{\setminus}C_{KM}|{=}|Q|$.

\input{algorithms/alg-MinimalCover}

If $M$ forms a minimal cover for $Q$, then $M$ is considered a query match.
However, $M$ may have some keyword matches which can be merged, especially SKMs. 
The merging of KMs from $M$ is carried out by Algorithm~\ref{alg:merge-keyword-matches}. 
Notice that we cannot merge two VKMs since they are disjoint sets, however we can merge a schema-keyword match with both a SKM or a VKM. 
The algorithm MergeKeywordMatches uses the two \emph{hash tables} $P_{V\!K}$ e $P_{S\!K}$ to store, respectively, the VKMs and SKMs based on the relation they are built upon (Lines \ref{line:mergekms-begin-hash-generation}-\ref{line:mergekms-end-hash-generation}). Next, the algorithm iterates through the relations present in $P_{SK}$ and tries to merge all possible KMs from that relation, resulting in a keyword match $KM_{merged}$. $KM_{merged}$ starts as a keyword-free match but it is merged with all existent SKMs (Lines \ref{line:mergekms-begin-skm-merge}-\ref{line:mergekms-end-skm-merge}), then it is merged an arbitrary value-keyword match $V\!K\!M$, if existent (Lines \ref{line:mergekms-begin-vkm-merge}-\ref{line:mergekms-end-vkm-merge}). Lastly, $KM_{merged}$ and all values-keyword matches except $V\!K\!M$ are added to the query match $M'$, which is returned at the end of MergeKeywordMatches (Lines \ref{line:mergekms-add-kmmerged}-\ref{line:mergekms-return}).

After merging all the possible elements from the query match $M$, QMGen adds $M$ to the set of query matches $Q\!M$, which is returned at the end of the algorithm.

\input{algorithms/alg-MergeQueryMatches}

%% file: algorithms/alg-QMGen.tex
\begin{algorithm}[!htb]
	\caption{QMGen($Q,V\!K,S\!K$)}
	\label{alg:qmgen}
	\KwIn{A keyword query $Q{=}\{k_{1},k_{2},\ldots,k_{m}\}$ \newline
		The set of value-keyword matches $V\!K$  \newline
		The set of schema-keyword matches $S\!K$ \newline
		The maximum QM size $t$\;
		
	}
	\KwOut{The set of query matches $Q\!M$}
	$P = V\!K \cup S\!K$\;
	$Q\!M \leftarrow \emptyset$\;
	\For{$i \in \{1,\ldots,\textbf{min}(|Q|,t)\}$}{ \label{line:qmgen-iteration-max-size}
	    \For{combination of keyword matches $M \in \binom{P}{i}$}{ \label{line:qmgen-iteration-combination}
			\If{\textbf{MinimalCover}($M,Q$)}{
				$M \leftarrow $\textbf{MergeKeywordMatches($M$)} \;
				$Q\!M\! \leftarrow Q\!M\! \cup \{M\}$ \;
			}
		}
	}
	\Return{$Q\!M$}	
\end{algorithm}

%% file: algorithms/alg-MinimalCover.tex
\begin{algorithm}[!htb]
	\caption{MinimalCover($Q,M$)}
	\label{alg:minimalCover}
	\KwIn{A keyword query $Q{=}\{k_{1},k_{2},\ldots,k_{m}\}$\newline
		The set of keyword matches $M$
	}
	\KwOut{If the set of keywords from $M$ forms a minimal cover over $Q$}
	$C_M \leftarrow \emptyset$ \; \label{line:minimalcover-begin-cover-qm}
	\For{keyword match $K\!M \in M$}{
	    \Let{$K\!M$ be $R^{X}[A_{1}^{K_{1}},\ldots,A_{m}^{K_{m}}]$,
	    where $X \in \{S,V\}$}
		\For{$i \in \{1,\ldots,m\}$}{
			$C_M \leftarrow C_M \cup K_{i}$\;
		}	
	} \label{line:minimalcover-end-cover-qm}
	\If{$ |C_M| \neq |Q|$}{ \label{line:minimalcover-check-total}
		\Return{$False$}
	}
	\For{keyword match $K\!M \in M$}{  \label{line:minimalcover-begin-cover-km}
	    \Let{$K\!M$ be $R^{X}[A_{1}^{K_{1}},\ldots,A_{m}^{K_{m}}]$,
	    where $X \in \{S,V\}$}
	    $C_{K\!M} = \emptyset$ \;
		\For{$i \in \{1,\ldots,m\}$}{
			$ C_{K\!M} \leftarrow C_{K\!M} \cup K_{i}$\;
		}	  \label{line:minimalcover-end-cover-km}
		\If{$|C_M \setminus C_{K\!M}| = Q$}{ \label{line:minimalcover-check-minimal}
			\Return{$False$}
		}
	}
	\Return{$True$}
\end{algorithm}

%% file: algorithms/alg-MergeQueryMatches.tex
 \begin{algorithm}[!htb]
	\caption{MergeKeywordMatches($Q,M$)}
	\label{alg:merge-keyword-matches}
	\KwIn{	The set of keyword matches $M$	}
	\KwOut{The set of keyword matches $M'$}
	\Let{$P_{V\!K}$ be a Hash Table.} \label{line:mergekms-begin-hash-generation}
	\Let{$P_{S\!K}$ be a Hash Table.}
	\For{ $K\!M \in M$}{
	    \Let{$K\!M$ be $R^{X}[A_{1}^{K_{1}},\ldots,A_{m}^{K_{m}}]$, where $X \in \{S,V\}$}
	    $P_{V\!K}[R] \leftarrow \emptyset$\;
	    $P_{S\!K}[R] \leftarrow \emptyset$\;
	}
	\For{ $K\!M \in M$}{
		\Let{$K\!M$ be $R^{X}[A_{1}^{K_{1}},\ldots,A_{m}^{K_{m}}]$,
	    where $X \in \{S,V\}$}
	     \uIf{$X = S$}{
            $P_{S\!K}[R] \leftarrow P_{S\!K}[R] \cup \{K\!M\}$\;
        }
        \Else{
            $P_{V\!K}[R] \leftarrow P_{V\!K}[R] \cup \{K\!M\}$\;
        }
	} \label{line:mergekms-end-hash-generation}
	$M' \leftarrow \emptyset$\;
	\For{$R \in P_{S\!K}$}{ \label{line:mergekms-begin-skm-merge}
	    \Let{$K\!M_{merged}$ be a keyword-free match from $R$}
	    \For{$S\!K\!M \in P_{S\!K}[R]$}{
	        $K\!M_{merged} \leftarrow K\!M_{merged} \cap S\!K\!M$\;
	    } \label{line:mergekms-end-skm-merge}
	    \If{$P_{V\!K}[R] \neq \emptyset$}{ \label{line:mergekms-begin-vkm-merge}
	        \Let{$V\!K\!M$ be an element from $P_{V\!K}[R]$}
	        $K\!M_{merged} \leftarrow K\!M_{merged} \cap V\!K\!M$\;
	        $P_{V\!K}[R] \leftarrow P_{V\!K}[R] - \{V\!K\!M\}$\;
	    } \label{line:mergekms-end-vkm-merge}
	    $M' \leftarrow M' \cup \{K\!M_{merged}\} \cup P_{V\!K}[R]$\; \label{line:mergekms-add-kmmerged}
	}
	\Return{$M'$} \label{line:mergekms-return}
\end{algorithm}

%% file: appendices/apx-d-qmrank.tex

    The ranking of Query Matches is carried out by Algorithm~\ref{alg:qmrank}, QMRank. Notice, that, intuitively, 
    the process of ranking QMS advances part of the relevance assessment of the CJNs, which was first proposed in CNRank~\cite{Oliveira@ICDE15CNRank}. This yields to an effective ranking of QMs and a simpler ranking of CJNs. QMRank uses a value score and a schema score, which are respectively related to the VKMs and SKMs that compose the QM.

\input{algorithms/alg-QMRank}

    The algorithm first iterates over each query match, assigning 1 to both $value\_score$ and $schema\_score$. Next, QMRank goes through each keyword match from the QM. In the case of a KM matching the values of an attribute, the algorithm updates the $value\_score$ based on the cosine similarity using TF-IDF weights. QMRank retrieves the term frequency and inverted attribute frequency from the Value Index , and the norm of an attribute from the Schema Index, which are all calculated in the preprocessing phase (see Section~\ref{sec:architecture}). 
    In the case of a KM matching the name of a schema element, the algorithm updates the $schema\_score$ the average similarity of the keywords with the schema elements based on the similarity functions presented in Section~\ref{chap:keyword-matching}.
    Once the algorithm aggregates the scores of KMs to generate the score of QMs, the final step is to sort them in descending order.

%% file: algorithms/alg-QMRank.tex
\begin{algorithm}[!htb]
	\caption{QMRank($Q\!M$)}
	\label{alg:qmrank}
	\KwIn{A set of query matches $Q\!M$
	}
	\KwOut{The set of ranked query matches $R\!Q\!M$}
	$R\!Q\!M \leftarrow [\ ]$\;
	\For{$M \in Q\!M$}{
		$value\_score \leftarrow 1,$
		$schema\_score \leftarrow 1$
		
		\For{ $K\!M \in M$}{
    		\Let{
    		    $K\!M$ be 
    		    $R^{S}[A_{1}^{K^S_{1}},\ldots,A_{m}^{K^S_{m}}]^{V}[A_{1}^{K^V_{1}},\ldots,A_{m}^{K^V_{m}}]$
    		}
    	    \For{$i \in \{1,\ldots,m\}$}{
    	        \If{$|K_{i}^V| \geq 1$}{
    	            $weight\_sum \leftarrow 0$\;
    	            $norm_{A_i} \leftarrow I_S[R][A_i]$\;
    	            \For{$word \in K_{i}^V$}{
    	                $tf \leftarrow |I_V[word][R][A_i]|$\;
    	                $weight\_sum \leftarrow weight\_sum + tf\times \textbf{iaf}(word)$\;
    	            }
					$value\_score \leftarrow value\_score \times weight\_sum/norm_{A_i}$
				}
				\If{$|K_{i}^S| \geq 1$}{
    	            $weight\_sum \leftarrow 0$\;
    	            \For{$word \in K_{i}^S$}{
    	                \uIf{$A_j = self$}{
							$schema\_element \leftarrow R$
						}
						\Else{
							$schema\_element \leftarrow A_i)$
						}
						$weight\_sum \leftarrow weight\_sum + \textbf{sim}(schema\_element,word)$\;
    	            }
					$schema\_score \leftarrow schema\_score \times weight\_sum/|K_{i}^S|$
				}
    	    }
    	}
		$final\_score \leftarrow value\_score \times schema\_score$\;
		$R\!Q\!M$.append($\langle final\_score, M \rangle$)
	}
	
	\textbf{Sort }$R\!Q\!M$ in descending order\; \label{line:qm-sorting}
	
	\Return{$R\!Q\!M$}	
\end{algorithm}

%% file: appendices/apx-e-soundtheorem.tex

\soundthm*

\begin{proof}
    \label{proof:sound-thm}
    Let $R_a$ and $R_b$ be database relations so that there exists $n$ \emph{ Referential Integrity Constraint} (RICs) from $R_a$ to $R_b$. Intuitively, a tuple from $R_a$ may refer to at most $n$ tuples from $R_b$.
    Consider a joining network of keyword matches $J$ wherein a keyword match over $R_a$ is adjacent to $m$ keyword matches over $R_b$, that is $J= \langle \mathcal{V}, E \rangle$, where $\mathcal{V}=\{K\!M_1,\ldots,K\!M_{m+1}\}$, $E = \{ \langle KM_1,KM_i \rangle | 2 \leq i \leq m\}$,
    and $R_1 = Ra \wedge R_i=R_b, 2 \leq i \leq m$.
    We can translate $J$ into a relational algebra expression wherein the edges are join operations using RICs and keyword matches are selection operations over relations.
    For didactic purposes, we assume, without loss of generality, that all the KMs of $J$ are keyword-free matches.
    Let $k_j$ be a key attribute from $R_i$ and $f_{i,j}$ be the attribute from $R_i$ that references $k_j$.
    The SQL translation of $J$ can be represented by $T_{m+1}$, which expands a join operation in each iteration.
    \begin{alignat*}{4}
	    &T_1 &=\ & R_1 &&\\
        &T_2 &=\ & T_1 \bowtie_{f_{1,2}=k_2} R_2 &&\\
        &T_3 &=\ & T_2 \bowtie_{f_{1,3}=k_3} R_3 &&\\
        &T_{n+1} &=\ & T_{n} \bowtie_{f_{1,{n+1}}=k_{n+1}} R_{n+1} &&\\
        &T_{n+2} &=\ & T_{n+1} \bowtie_{f_{1,{x}}=k_{n+2}} R_{n+2}, &&\text{where } x \in \{2,\ldots,n+1\} &&\\
	\end{alignat*}
	Notice that by the iteration $n+2$, all RICs from $R_a$ to $R_b$ were already used once. Therefore, this expansion require that we use one of the RICs twice, which would lead to redundancy.
	For instance, if assume $x=2$, without loss of generality, then:
	\begin{alignat*}{4}
        &T_2 &=\ & T_1 \bowtie_{f_{1,2}=k_2} R_2 &&\\
        &T_{n+2} &=\ & T_{n+1} \bowtie_{f_{1,{2}}=k_{n+2}} R_{n+2} &&\\
	\end{alignat*}
	As the join conditions are stacked in each iteration, we can say that:
	\[
	    f_{1,2}=k_2 \wedge f_{1,{2}}=k_{n+2} 
	\]
	which implies that $k_2 = k_{n+2}$ and, thus, all the returning JNTs would have more than one occurrence of the same tuple for every instance of the database.
	
	 \begin{alignat*}{2}
        &T_{m+1} &=\ & T_m \bowtie_{f_{1,{x}}=k_{m+1}} R_{m+1}\\
	\end{alignat*}
\end{proof}

%% file: appendices/apx-f-cnkmgen.tex

The generation and ranking of CJNs is carried out by Algorithm~\ref{alg:cnkmgen}, CJNGen, which uses a \textit{Breadth-First Search} approach \cite{Cormen@09Algorithms} to expand JNKMs until they comprehend all elements from a query match.

Despite being based on the MatCNGen Algorithm~\cite{Oliveira@ICDE18MatCNGen}, CJNGen provides support for generating CJNs wherein there exists more than one RIC between one database relation to another, due to the definition of soundness presented in Theorem~\ref{thm:sound-cjn}. Also, CJNGen does not require an intermediate structure such as the \emph{Match Graph} in the MatCNGen system.

We describe CJNGen in Algorithm~\ref{alg:cnkmgen}. For each query match, CJNGen generates the candidate joining networks for this query match using an internal algorithm called CJNInter, which we will focus on describing in the remainder of this section.

\input{algorithms/alg-SVMatCNGen}

\input{algorithms/alg-CNKMIter}

In Algorithm~\ref{alg:cnkmperqmgen}, we present CJNInter. This algorithm takes as input a query match $M$ and the schema graph $G_S$. Next, it chooses a KM from the QM as a starting point, resulting in an unitary graph (Lines~\ref{line:cnkm-starting-node-start}-\ref{line:cnkm-starting-node-end}). If the query match $M$ has only one element, we already generated the one possible candidate joining network (Line~\ref{line:cnkm-unitary-cn}).

Next, the CJNInter initializes a queue $D$, which is used to store the JNKMs which are not CJNs (Lines~\ref{line:cnkm-inicialization-start}-\ref{line:cnkm-inicialization-end}). 
In Loop~\ref{line:cnkm-begin-expansion}-\ref{line:cnkm-end-expansion}, CJNInter takes one JNKM $J$ from the queue and tries to expand it with KMs. N
otice that $J$ can be expanded with incoming and outgoing neighbors, therefore it uses an undirected schema graph $G_S^U$(Line~\ref{line:cnkm-undirected-schema-graph}). 
Also, the elements of $M$ can only be added once in a JNKM but keyword-free matches can be added several times.

The expansion of $J$ results in a JNKM $J'$ (Lines~\ref{line:cnkm-new-tree-start}-\ref{line:cnkm-new-tree-end}). Then , CJNInter verifies whether $J'$ was already generated and whether it is \textit{sound}, according to Definition~\ref{def:sound-candidate-network}. If $J'$ fails to meet these two conditions it is pruned (Line~\ref{line:cnkm-check-sound}).

If $J'$ was not pruned, CJNInter checks whether $J'$ covers the query match $M$. If it does, $J'$ is a candidate joining network and it will be added to the list $C\!J\!N$. If $J'$ does not cover $M$, then it will  be added to the deque $D$ (Lines~\ref{line:cnkm-check-become-cn-start} -\ref{line:cnkm-check-become-cn-end}).
At the end of the procedure, CJNInter returns the set $C\!J\!N$ of candidate joining networks for the query match $M$(Line~\ref{line:cnkm-return-results}).

The CJN generation algorithm also implements some basic CJN pruning strategies, which are based on the following parameters: the top-k CJNs, the top-k CJNs per QM and the maximum CJN size.
Also, the algorithm implements a few strategies to prune the JNKMs which are not minimal or not sound, the maximum node Degree, the maximum number of keyword-free matches, and the distinct foreign keys.

\subsection{Maximum Node Degree}

As the leaves of a CJN must be keyword matches from the query match, then a CJN must have at most $|Q\!M|$ leaves. Also, considering that the maximum node degree in a tree is less or equal to the number of its leaves, we can safely prune the JNKMs that contains a node with a degree greater than $|Q\!M|$.

\subsection{Maximum Number of Keyword-free Matches}

The size of a CJN is based on the size of the query match and the number of keyword-free matches, that is, the size of a candidate joining network $C\!J\!N_M$ for a query match $M$ is given by $|C\!J\!N_M| {=} |M|{+}|F|$, where $F$ is a set of keyword-free matches. Thus, if we consider a maximum CJN size $T_{max}$, we can also set a maximum number of keyword-free matches for a CJN, given by $|F| {\leq} T_{max}{-}|M|$. Therefore, we can prune all JNKMs that contain more keyword-free matches than this maximum number set.

The number of CJNs generated can be further reduced by the pruning and ranking them.
In Section~\ref{sec:cn-ranking}, we present a ranking of the candidate joining networks returned by CJNGen. In Section~\ref{sec:cn-pruning}, we present pruning techniques for the generation of the candidate joining networks from CJNGen and CJNInter.

%% file: algorithms/alg-SVMatCNGen.tex
\begin{algorithm}[htb]
	\caption{CJNGen($R\!Q\!M,G_S$)}
	\label{alg:cnkmgen}
	\KwIn{The set of ranked query matches $R\!Q\!M$  \newline
		The schema graph $G_S$
	}
	\KwOut{The set of candidate networks $C\!J\!N$}
	
	$C\!J\!N = \{\}$\;
	\For{query match $M \in R\!Q\!M$}{
		$C\!J\!N_M \leftarrow$ \textbf{CJNInter}($M,G_S$)\;	
		$C\!J\!N \leftarrow C\!J\!N \cup C\!J\!N_M$
	}
	\Return{$C\!J\!N$}
\end{algorithm}

%% file: algorithms/alg-CNKMIter.tex
\begin{algorithm}[!htb]
\caption{CJNInter($G_S,M,score_{M}$)}

\label{alg:cnkmperqmgen}

\KwIn{The query match $M$; The schema graph $G_S$
}

\KwOut{A set $C\!J\!N$ of candidate networks for the query match $M$}

$C\!J\!N \leftarrow [\ ]$\;
$J \leftarrow$ Graph() \;
\Let{$K\!M$ be an element from $M$} \label{line:cnkm-starting-node-start}

Add $K\!M$ to $J$.$\mathcal{V}$ \; \label{line:cnkm-starting-node-end}

\If{$|M|=1$}{ 
	\Return{$\{J\}$} \label{line:cnkm-unitary-cn}
}
$D \leftarrow$ queue() \; \label{line:cnkm-inicialization-start}
$D\!$.enqueue($J$) \; \label{line:cnkm-inicialization-end}
\While{$D \neq \{\}$}{ \label{line:cnkm-begin-expansion} 
	$J \leftarrow D\!$.dequeue()\;
	\For{$K\!M_u \in J.\mathcal{V}$}{
		\Let{
		    $K\!M_u$ be 
		    $R_u^{S}[A_{u,1}^{K^S_{u,1}},\ldots,A_{u,m}^{K^S_{u,m}}]^{V}[A_{u,1}^{K^V_{u,1}},\ldots,A_{u,m_u}^{K^V_{u,m_u}}]$
		}
		\Let{$G_S^U$ be the undirected version of $G_S$}
		\For{ $R_a$ adjacent to $R_u$ in $G_S^U$}{ \label{line:cnkm-undirected-schema-graph} 
			\For{$K\!M_v \in M{\setminus}C\!N\!.\mathcal{V}$}{
			    \Let{
        		    $K\!M_v$ be 
        		    $R_v^{S}[A_{v,1}^{K^S_{v,1}},\ldots,A_{v,m}^{K^S_{v,m}}]^{V}[A_{v,1}^{K^V_{v,1}},\ldots,A_{v,m_v}^{K^V_{v,m_v}}]$
        		}
			    \If{$R_v = R_a$}{
				    $J' \leftarrow J$\;				\label{line:cnkm-new-tree-start} 
    				Expand $J'$ with $K\!M_v$ joined to $K\!M_u$\; \label{line:cnkm-new-tree-end} 
    				
    				\If{$J' \notin C\!N$ \KwAnd $J'$ is \emph{sound}}{ \label{line:cnkm-check-sound} 
    					\If{$J'.\mathcal{V} \supseteq M$}{   \label{line:cnkm-check-become-cn-start} 
    						$C\!N$.append($J'$) \;  
    					}
    					\Else{
    						$D$.enqueue($J'$) \; \label{line:cnkm-check-become-cn-end}
    					}
    				}
				}
			}
			$J' \leftarrow J$\;		
			Expand $J'$ with $R_a^S[\ ]^V[\ ]$ joined to $K\!M_u$\;
			$D$.enqueue($J'$) \;
		}
	}
} \label{line:cnkm-end-expansion}
\Return{$C\!J\!N$}  \label{line:cnkm-return-results}
\end{algorithm}